\documentclass{eptcs-modified}
\usepackage{tikz}

 \providecommand{\event}{} 
\usepackage{amsmath}
\usepackage{amssymb}
\usepackage{amsthm}
\usepackage{amscd}
\usepackage{graphics}

\title{Computable dyadic subbases and $\mathbf{T}^\omega$-representations of compact sets}
\author{
Arno Pauly
\institute{Department of Computer Science, Swansea University, Swansea, UK\\
Birmingham University, United Kingdom}
\email{Arno.M.Pauly@gmail.com}
\and
Hideki Tsuiki
\institute{Graduate School of Human and Environmental Studies,\\ Kyoto University, Kyoto, Japan}
\email{tsuiki@i.h.kyoto-u.ac.jp}
}
\def\titlerunning{$\mathbf{T}^\omega$ representations from dyadic subbases}
\def\authorrunning{A. Pauly \& H. Tsuiki}

\begin{document}
\theoremstyle{definition}
\newtheorem{theorem}{Theorem}
\newtheorem{definition}[theorem]{Definition}
\newtheorem{problem}[theorem]{Problem}
\newtheorem{assumption}[theorem]{Assumption}
\newtheorem{corollary}[theorem]{Corollary}
\newtheorem{proposition}[theorem]{Proposition}
\newtheorem{lemma}[theorem]{Lemma}
\newtheorem{observation}[theorem]{Observation}
\newtheorem{fact}[theorem]{Fact}
\newtheorem{question}[theorem]{Open Question}
\newtheorem{conjecture}[theorem]{Conjecture}
\newtheorem{example}[theorem]{Example}
\newcommand{\dom}{\operatorname{dom}}
\newcommand{\id}{\textnormal{id}}
\def\2{\{0,1\}}
\newcommand{\Cantor}{{\2^\omega}}
\newcommand{\Baire}{{\mathbb{N}^\omega}}
\newcommand{\Lev}{\textnormal{Lev}}
\newcommand{\hide}[1]{}
\newcommand{\mto}{\rightrightarrows}
\newcommand{\uint}{{[0, 1]}}
\newcommand{\bft}{\mathrm{BFT}}
\newcommand{\lbft}{\textnormal{Linear-}\mathrm{BFT}}
\newcommand{\pbft}{\textnormal{Poly-}\mathrm{BFT}}
\newcommand{\sbft}{\textnormal{Smooth-}\mathrm{BFT}}
\newcommand{\ivt}{\mathrm{IVT}}
\newcommand{\cc}{\textrm{CC}}
\newcommand{\lpo}{\textrm{LPO}}
\newcommand{\llpo}{\textrm{LLPO}}
\newcommand{\aou}{AoU}
\newcommand{\Ctwo}{C_{\{0, 1\}}}
\newcommand{\name}[1]{\textsc{#1}}
\newcommand{\C}{\textrm{C}}
\newcommand{\UC}{\textrm{UC}}
\newcommand{\ic}[1]{\textrm{C}_{\sharp #1}}
\newcommand{\xc}[1]{\textrm{XC}_{#1}}
\newcommand{\me}{\name{P}.~}
\newcommand{\etal}{et al.~}
\newcommand{\eval}{\operatorname{eval}}
\newcommand{\rank}{\operatorname{rank}}
\newcommand{\Sierp}{Sierpi\'nski }
\newcommand{\isempty}{\operatorname{IsEmpty}}
\newcommand{\spec}{\textrm{Spec}}
\newcommand{\cord}{\textrm{COrd}}
\newcommand{\Cord}{\textrm{\bf COrd}}
\newcommand{\CordM}{\Cord_{\textrm{M}}}
\newcommand{\CordK}{\Cord_{\textrm{K}}}
\newcommand{\CordHL}{\Cord_{\textrm{HL}}}
\newcommand{\leqW}{\leq_{\textrm{W}}}
\newcommand{\leqsW}{\leq_{\textrm{sW}}}
\newcommand{\leW}{<_{\textrm{W}}}
\newcommand{\equivW}{\equiv_{\textrm{W}}}
\newcommand{\geqW}{\geq_{\textrm{W}}}
\newcommand{\pipeW}{|_{\textrm{W}}}
\newcommand{\nleqW}{\nleq_{\textrm{W}}}
\newcommand{\Det}{\textrm{Det}}
\newcommand{\R}{\textrm{R}}
\newcommand{\UR}{\textrm{UR}}
\newcommand{\bo}[2]{{{#1_{{\langle#2\rangle}}}}}
\newcommand{\Q}{\mathbb{Q}}
\def\T{\mathbb{T}}
\def\A{\mathcal{A}}
\def\K{\mathcal{K}}
\def\V{\mathcal{V}}
\def\O{\mathcal{O}}
\def\X{\mathbf{X}}
\def\Y{\mathbf{Y}}
\def\N{\mathbb{N}}
\newcommand{\pruned}{\mathcal{PT}}
\newcommand{\tree}{\mathcal{T}}
\newcommand{\treeL}{\mathcal{T}_\mathcal{L}}
\newcommand{\comp}{\uparrow}
\def\exS{S_{\mathrm{ex}}}
\def\exbarS{\bar{S}_{\mathrm{ex}}}
\newcommand{\cl}{\textrm{cl}}
\newcommand{\Real}{\mathbb{R}}
\newcommand{\extr}{\mathrm{ext}}

\newcommand\tboldsymbol[1]{%
\protect\raisebox{0pt}[0pt][0pt]{%
$\underset{\widetilde{}}{\boldsymbol{#1}}$}\mbox{\hskip 1pt}}

\newcommand{\bolds}{\tboldsymbol{\Sigma}}
\newcommand{\boldp}{\tboldsymbol{\Pi}}
\newcommand{\boldd}{\tboldsymbol{\Delta}}
\newcommand{\boldg}{\tboldsymbol{\Gamma}}

\newcounter{saveenumi}
\newcommand{\seti}{\setcounter{saveenumi}{\value{enumi}}}
\newcommand{\conti}{\setcounter{enumi}{\value{saveenumi}}}

\maketitle

\begin{abstract}    
We explore representing the compact subsets of a given represented space by infinite sequences over Plotkin's $\T$. We show that computably compact computable metric spaces admit
representations of their compact subsets
in such a way that compact sets are essentially underspecified points.
We can even ensure that a name of an $n$-element compact set contains $n$ 
 occurrences of $\bot$.
We undergo this study effectively  and show that
such a $\T^\omega$-representation is effectively obtained from structures of
computably compact computable metric spaces.
As an application, we prove some statements about the Weihrauch degree of closed choice for finite subsets of computably compact computable metric spaces.

Along the way, we introduce the notion of a computable dyadic subbase, and prove that every computably compact computable metric space admits a proper computable dyadic subbase.
\end{abstract}

\maketitle


\section{Introduction}
In TTE \cite{weihrauchd}, the fundamental computability notion is introduced on either Cantor space $\Cantor$ or Baire space $\Baire$, and then lifted to other spaces of interest via representations. It is well-known that the choice of $\Cantor$ or $\Baire$ is inconsequential for the resulting theory, and authors typically choose whatever space works better for a specific purpose. In principle, other spaces can be used as the fundament, too, provided that they have a sufficiently substantial computability theory defined on them. Using the space of regular word functions has been advocated by \name{Kawamura} and \name{Cook} with computational complexity as the motivation \cite{kawamura}. If one is primarily interested in Quasi-Polish spaces \cite{debrecht6}, then the Scott domain $\mathcal{P}(\omega)$ makes sense as the foundational space, with a computability notion derived from enumeration reducibility (cf.~\cite{pauly-kihara-arxiv}).

Here we consider $\T^\omega$ as a fundamental space for computation, the space of infinite sequences over Plotkin's $\T$. Plotkin's $\T$ is the three point space $\{0,1,\bot\}$ with the topology generated by $\{\{0\},\{1\}\}$. Thus, $\bot$ plays the role of \emph{not yet determined}, whereas the values $0$ and $1$, once attained, will remain unchanged. The use of $\T^\omega$ (together with IM2 machines) as the basis for a theory of computability has been investigated by the second author in a number of papers \cite{tsuiki,tsuiki4,tsuiki3}. An interesting result is that a separable
metric space $\mathbf{X}$ admits an injective representation $\delta : \subseteq \T^\omega \to \mathbf{X}$ such that each $p \in \dom(\delta)$ has at most $n$ occurrences of $\bot$ iff the dimension of $\mathbf{X}$ is at most $n$.

This approach appears very similar to the domain-representations studied in e.g.~\cite{blanck}. The significant difference is that we allow for multivalued realizers,
and thus never change the resulting notion of computability. For domain-representations, realizers are required to be 
single-valued, and the resulting categories can thus differ fundamentally from the category of represented spaces and computable functions.

In the present paper we consider $\T^\omega$-representations of the space $\mathcal{K}(\mathbf{X})$ of compact subsets of some space $\mathbf{X}$ represented over $\Cantor$. We first introduce some notions on representations of non-empty compact subsets.
We are particularly interested in matching representations in the following sense:

\begin{definition}\label{def:match}
Consider a representation $\delta :\subseteq \Cantor \to \mathbf{X}$ and a $\T^\omega$-representation $\psi :\subseteq \T^\omega \to \mathcal{K}(\mathbf{X})\setminus \{\emptyset\}$. We say that they \emph{match} iff
\[
    \psi(p) = \{\delta(q) \mid p \preceq q \in \dom(\delta)\}
\]
for every $p \in \dom(\psi)$.
Here $\preceq$ denotes the specialization relation on $\T^\omega \supset \Cantor$
(see Section \ref{section-notation} for details on notation).
\end{definition}


A pair of matching representations means that we can obtain names for points inside the compact set by replacing $\bot$ with $0$ or $1$.
This essentially means that we can consider compact subsets as underspecified points. This seems like a counterpart to the identification of points in admissible spaces as being equivalent to compact singletons \cite{schroder5,pauly-synthetic}.
We consider more: we can obtain names for compact subsets
by replacing some copies of $\bot$ with $0$ or $1$.

\begin{definition}\label{def:hereditary}
A $\T^\omega$-representation $\psi :\subseteq \T^\omega \to \mathcal{K}(\mathbf{X}) \setminus \{\emptyset\}$ is \emph{hereditary} if the restriction of $\psi$ to $\uparrow\!\! p \cap \dom(\psi)$ is a
representation of $\uparrow\!\psi(p) \subseteq \mathcal{K}(\mathbf{X}) \setminus \{\emptyset\}$
for every $p \in \dom(\psi)$.
\end{definition}
Here,  $\uparrow\!a = \{x \mid x \succeq a\}$ is the upper set with respect to the specialization order $\preceq$,  which for $\mathcal{K}(\mathbf{X})$ is  $A \preceq B$ iff $B \subseteq A$.
If $\psi$ is a hereditary representation and $p \prec q$, then $\psi(p) \supseteq \psi(q)$.
Therefore, every name of a $n$-point set $A$ contains at least $n-1$ copies of $\bot$ because there is a chain of length $n-1$ in $\mathcal{K}(A) \setminus \{\emptyset\}$.
We want even more for our applications.
\begin{definition}\label{def:faithful}
A $\T^\omega$-representation $\psi :\subseteq \T^\omega \to \mathcal{K}(\mathbf{X}) \setminus \{\emptyset\}$ is \emph{faithful} if it is hereditary and the followings properties hold for the case $\psi(p)$ is a finite subset:
\begin{itemize}

\item[(1)] $p$ contains $|\psi(p)| - 1$ copies of $\bot$.
\item[(2)] $\uparrow p \cap \{0,1\}^\omega \subseteq \dom(\psi)$.
\end{itemize}
\end{definition}

Since the restriction of a faithful  $\T^\omega$-representation $\psi$ of $\K(\X) \setminus \{\emptyset\}$
to $\dom(\psi) \cap \{0,1\}^\omega $ is a representation of singleton sets,
it induces a representation $\delta$ of $\X$.
The pair $(\delta, \psi)$ obviously forms a matching representation.

From a $\T^\omega$-representation $\psi' :\subseteq \T^\omega \to \mathcal{K}(\mathbf{X})\setminus \{\emptyset\}$, one can easily form a $\T^\omega$-representation
$\psi :\subseteq \T^\omega \to \mathcal{K}(\mathbf{X})$ by defining $\psi(\bot p) = \psi'(p)$ and
$\psi(1 p) = \psi(0 p) = \emptyset$ for $p \in \T^\omega$.
Therefore, we extend the above notions to those of $\T^\omega$-representations of $\mathcal{K}(\mathbf{X})$ as follows.
\begin{definition}\label{def:match2}
A $\T^\omega$-representation $\psi :\subseteq \T^\omega \to \mathcal{K}(\mathbf{X})$  is faithful
iff,  $\psi(0p) = \psi(1p) = \emptyset$ and the $\T^\omega$-representation $\psi'$ of
$\mathcal{K}(\mathbf{X}) \setminus \{\emptyset\}$ defined as $\psi'(p) =\psi(\bot p)$ is
faithful.  Similary, we define a hereditary $\T^\omega$-representation of $\mathcal{K}(\mathbf{X})$
and a $\T^\omega$-representation of $\mathcal{K}(\mathbf{X})$ that
matches with a representation $\delta$ of $\X$.
\end{definition}
If $\psi$ is a faithful $\T^\omega$-representation of $\K(\X)$, then
every name of a $n$-point set contains $n$ copies of $\bot$.

We provide a faithful $\T^\omega$-representation for the case $\X$ being the Cantor space $\Cantor$ in Section \ref{sec:cantor}.
For more general spaces,
we  provide a construction of a faithful $\T^\omega$-representation
based on the concept of a proper dyadic subbase.
We study computable dyadic subbases
and  show that every computable metric space has a proper computable dyadic subbase by effectivizing a proof of \name{Tsukamoto} \cite{tsukamoto} in Section
\ref{sec:propersubbases}.
Then, in Section \ref{sec:mainresult},
we show that a computably compact computable metric space with a proper computable dyadic subbase has a faithful $\T^\omega$-representation.

  We represent compact subsets as trees using proper computable dyadic subbases,
and then construct a $\T^\omega$-representation of trees.
  Tree-based representations of compact sets appeared in \cite{berger}, where they used the Vietoris topology (aka the Hausdorff metric aka the full information representation) on the hyper space.  Our representation is linked to the upper Vietoris topology, which is compatible with $\T^\omega$-representations.

An initial motivation of this work was the investigation of a construction employed to prove \cite[Proposition 1.9]{paulyleroux}. That result is showing that the Weihrauch degree of finding a point in a closed subset of a computably rich computably compact computable metric space $\X$ with cardinality equal to $n$ (or up to $n$) does not depend on the choice of $\X$. The construction in \cite{paulyleroux} crucially depends on knowing $n$. Here, we can prove that the statement remains true without any bound on the finite cardinality. In Section \ref{sec:weihrauch}, we explore the resulting Weihrauch degree and relate it to other known degrees.

\section{Notation and Fundamentals}

\subsection{Notation}\label{section-notation}
For an alphabet $\Sigma$, we denote by $\Sigma^*$ respectively $\Sigma^\omega$ the space of
finite respectively infinite sequences over $\Sigma$.
For $w, v \in \Sigma^*$, let $w \sqsubseteq v$ express that $w$ is a prefix of $v$.
For $w \in \Sigma^*$, we let $|w| \in \mathbb{N}$ denote its length.
For $p \in {\Sigma}^\omega$ or $p \in {\Sigma}^*$, we denote by $p_{\leq n} \in {\Sigma}^n$ its restriction to first $n$ components.  Here, in the case $p \in {\Sigma}^*$, we assume that $n \leq |p|$.
As a special case, $\T^*$ is the space of finite sequences over $\T=\{0,1,\bot\}$.
We consider $\iota : \T^* \to \T^\omega$ mapping $e \in \T^*$ to $e\bot^\omega$ as a standard computable map, but point out that the partial inverse of $\iota$ is not computable
\footnote{Note that in the previous literature on proper dyadic subbases the expression $\T^*$ denotes our
$\iota(\T^*)$.}.

We call the number of digits (i.e., 0 or 1) in  $e \in \T^*$ the level of $e$
and denote it by $\mathrm{level}(e)$.
We denote by $\bo{\T^*}{n} \subset \T^*$ the set of level-$n$ sequences.
More generally, for a subset $A$ of $\T^*$,
we denote by $\bo{A}{n}$ the set of level-$n$ sequences of $A$.
For example,
$1 0 1, 1 \bot 1, 1\bot\bot, \bot\bot\bot$ all have length 3 and
belong to ${\bo{\T^*}{3}}, {\bo{\T^*}{2}}, {\bo{\T^*}{1}}, {\bo{\T^*}{0}}$, respectively.
We write $\bo{A}{\leq n} := \bigcup_{i \leq n} \bo{A}{i}$.
For $p \in \T^*$ or $p \in \T^\omega$, we define $\bo{p}{\leq n} \in \bo{\T^*}{n}$ as
$p_{\leq m+1}$
for $m$ the index of the $n$-th digit of $p$. For example, $\bo{\bot 01}{\leq 1} = \bot0$.

We write $a \preceq b$ for $a, b \in \T$ if $a = \bot \lor a = b$, and
$p \preceq q$ for $p, q \in \T^\omega$  if
$\forall n \in \mathbb{N} \ p(n) \preceq q(n)$.

 We write $\dom(p) = \{n \in \mathbb{N} \mid p(n) \neq \bot\}$. By $p \comp q$ we denote that $\forall n \in \dom(p) \cap \dom(q) \ p(n) = q(n)$. We extend these notions to $\T^*$ along $\iota$. Note that $\preceq$ is a partial order on $\T^\omega$ but is only a quasiorder on $\T^*$, as e.g.~$0 \preceq 0\bot$ and $0\bot \preceq 0$. By excluding from $e \in \T^*$ the finite sequence ending in $\bot$ we obtain canonic representative $\tau(e)$ of each $\preceq$-equivalence class.  $\tau(\T^*)$ is a partially ordered set.

By $\mathbf{3}$ we denote the set $\{0,1,\bot\}$ equipped with the representation $\delta_\mathbf{3}(10^\omega) = 0$, $\delta_\mathbf{3}(110^\omega) = 1$ and $\delta_\mathbf{3}(1110^\omega) = \bot$.
With the representation we will introduce in Section \ref{section-introtomega},
$\id : \mathbf{3} \to \mathbb{T}$ is computable, but $\id : \mathbb{T} \to \mathbf{3}$ is not. In the following, we will suppress both $\iota : \T^* \to \T^\omega$ and $\id : \mathbf{3} \to \mathbb{T}$ and their combinations in the notation. For example, if we have some function $f : \mathbb{T}^\omega \to \mathbf{X}$, we might speak of the function $f : \mathbf{3}^* \to \mathbf{X}$ obtained by precomposing with these computable functions without further notice.

\subsection{Background on represented spaces}
We briefly recall some fundamental concepts on represented spaces following \cite{pauly-synthetic}, to which the reader shall also be referred for a more extensive treatment. A \emph{represented space} is a pair $\mathbf{X} = (X, \delta_X)$ of a set $X$ and a partial surjection $\delta_X : \subseteq \Cantor \to X$. A (multivalued) function between represented spaces is a (multivalued) function between the underlying sets. For $f : \mathbf{X} \mto \mathbf{Y}$ and $F : \subseteq \Cantor \to \Cantor$, we call $F$ a realizer of $f$ (notation $F \vdash f$), iff $\delta_Y(F(p)) = f(\delta_X(p))$ for all $p \in \dom(f\delta_X)$, i.e.~if the following diagram commutes:
 $$\begin{CD}
\Cantor @>F>> \Cantor\\
@VV\delta_\mathbf{X}V @VV\delta_\mathbf{Y}V\\
\mathbf{X} @>f>> \mathbf{Y}
\end{CD}$$
A map between represented spaces is called computable (continuous), iff it has a computable (continuous) realizer. A priori, the notion of a continuous map between represented spaces and a continuous map between topological spaces are distinct. However, for the admissible spaces (as defined by \name{Schr\"oder} \cite{schroder,schroder5}) the notions coincide. This in particular includes the computable metric spaces.

We say that representations $\delta_1$, $\delta_2$ of the same set $X$ are equivalent, if $\id : (X,\delta_1) \to (X,\delta_2)$ and $\id : (X,\delta_2) \to (X,\delta_1)$ are both computable. In this case, we also say that $\delta_2$ is a representation of $(X,\delta_1)$, or that $(X,\delta_1)$ admits the representation $\delta_2$.

Cantor space itself is considered a represented space, with $\id : \Cantor \to \Cantor$ serving as the representation. Other specific relevant represented spaces are $\mathbb{N}$ represented by $\delta_\mathbb{N}(0^n1^\omega) = n$ and Sierpi\'nski space $\mathbb{S}$ with underlying set $\{\bot, \top\}$ and representation $\delta_\mathbb{S}$ defined by $\delta_\mathbb{S}(0^\omega) = \bot$ and $\delta_\mathbb{S}(p) = \top$ for $p \neq 0^\omega$.

Courtesy of the UTM-theorem (in the form proven in \cite{weihrauchk}), there is a natural representation of the space $\mathcal{C}(\mathbf{X},\mathbf{Y})$ of continuous functions between two given represented spaces $\mathbf{X}$, $\mathbf{Y}$. Then representation is rendering all the expected operations computable, in particular composition and evaluation. We immediately obtain a representation of any $\mathbf{X}^\omega$ in form of $\mathcal{C}(\mathbb{N},\mathbf{X})$. We can also derive the space $\mathcal{O}(\mathbf{X})$ of open subsets of $\mathbf{X}$ by identifying a set $U \subseteq \mathbf{X}$ with its characteristic function $\mathcal{C}(\mathbf{X},\mathbb{S}) \ni \chi_U : \mathbf{X} \to \mathbb{S}$ mapping $x \in U$ to $\top$ and $x \notin U$ to $\bot$. The open subsets are the final topology along the representation, and again, the expected operations are computable.

The represented space $\mathcal{A}(\mathbf{X})$ of closed subsets is defined
by considering the characteristic function of its complement. In other words, we define $\mathcal{A}(\mathbf{X})$ in such a way that $^C : \mathcal{O}(\mathbf{X}) \to \mathcal{A}(\mathbf{X})$ and $^C : \mathcal{A}(\mathbf{X}) \to \mathcal{O}(\mathbf{X})$ became computable.
We further introduce the space $\mathcal{K}(\mathbf{X})$ of compact subsets by representing $A \subseteq \mathbf{X}$ via $\{U \in \mathcal{O}(\mathbf{X}) \mid A \subseteq U\} \in \mathcal{O}(\mathcal{O}(\mathbf{X}))$.
Dually, we define the represented space $\mathcal{V}(\mathbf{X})$ of
overt subsets by representing $A \subseteq \mathbf{X}$ via $\{U \in \mathcal{O}(\mathbf{X}) \mid A \cap U \ne \emptyset\} \in \mathcal{O}(\mathcal{O}(\mathbf{X}))$. Note that the elements of $\mathcal{V}(\mathbf{X})$ individually are closed sets, but neither $\id : \mathcal{A}(\mathbf{X}) \to \mathcal{V}(\mathbf{X})$ nor $\id : \mathcal{V}(\mathbf{X}) \to \mathcal{A}(\mathbf{X})$ is computable for non-empty $\mathbf{X}$.

A represented space $\mathbf{X}$ is called computably compact if $\textrm{isEmpty} : \mathcal{A}(\mathbf{X}) \to \mathbb{S}$ is computable,
and computably overt if $\textrm{IsNonEmpty}: \mathcal{O}(\mathbf{X}) \to \mathbb{S}$ is computable.
It is called computably Hausdorff, iff $\mathalpha{\neq} : \mathbf{X} \times \mathbf{X} \to \mathbb{S}$ is computable. A space is computably compact and computably Hausdorff iff both $\id : \mathcal{A}(\mathbf{X}) \to \mathcal{K}(\mathbf{X})$ and $\id : \mathcal{K}(\mathbf{X}) \to \mathcal{A}(\mathbf{X})$ are well-defined and computable. As we will be working with (computable) compact Hausdorff spaces, we can freely alternate between treating sets represented as closed or as compact sets in the following.

\subsection{Computably compact computably Hausdorff countably-based spaces}
The proofs of our main results will make explicit use of the space $\X$ being computably compact and being computably Hausdorff. Moreover, since we are using dyadic subbases, the space has to be countably-based. It is a classical result in topology that countably based compact Hausdorff spaces are metrizable. This is shown in two steps: First, it is shown that compact Hausdorff spaces are regular, then Urysohn's metrization theorem tells us that countably-based regular Hausdorff spaces are metrizable.

In an effective context, Schr\"oder's effective metrization theorem \cite{grubba3,schroder8} shows that countably-based computably-regular computably Hausdorff spaces are computably metrizable. However, the standard proof that compact Hausdorff implies regular does not effectivize:

The argument proceeds as follows: We have a point $x \in \X$ and a (non-empty) closed set $A \in \mathcal{A}(\X)$ with $x \notin A$. For every $y \in A$ there exist disjoint open $U_y, V_y \in \mathcal{O}(\mathbf{X})$ with $x \in U_y$ and $y \in V_y$, since $\X$ is Hausdorff. Now $\{V_y \mid y \in A\}$ is an open cover of $A$. By compactness, it has a subcover indexed by a finite set $I \subseteq A$. Now $V = \bigcap_{y \in I} V_y$ and $U = \bigcup_{y \in I} U_y$ are disjoint open sets with $x \in V$ and $A \subseteq U$, establishing that $\X$ is regular. The problem is that to effectively use the open cover $\{V_y \mid y \in A\}$, we would need to know $A$ as an overt set, not as a closed set.

Whether there is a different proof that establishes the full effective version of this result has been raised as \cite[Question 9]{oberwolfach-computability}. However, we only need it in the countably-based case. Here, we have effectively open representations available, meaning that $U \mapsto \{\delta(p) \mid p \in U \cap \dom(\delta)\} : \mathcal{O}(\Cantor) \to \mathcal{O}(\mathbf{X})$ is well-defined and computable.

\begin{theorem}\label{theorem:coregular}
Let $\mathbf{X}$ admit an effectively open representation, be computably Hausdorff and computably compact. Then $\mathbf{X}$ is computably regular.
\begin{proof}
We are given $x \in \mathbf{X}$ and $A \in \mathcal{A}(\X)$ with $x \notin A$, and we need to compute $U, V \in \mathcal{O}(\mathbf{X})$ with $U \cap V = \emptyset$, $x \in U$ and $A \subseteq V$. We use computable compactness of $\X$ to obtain $A \in \mathcal{K}(\X)$. Since $\X$ is computably Hausdorff, there exists a computable function $N : \subseteq \Baire \times \Baire \to \mathbb{S}$ such that $N \circ \langle \delta_\X, \delta_\X\rangle$ is $\mathalpha{\neq} : \X \times \X \to \mathbb{S}$. We can extend $N$ to a total computable function, as $\mathbb{S}$ is precomplete.

We exhaustively search for finite prefixes $(v_w,w)$ such that the realizer of $N$ writes a $1$ somewhere when reading in a prefix $v_w$ of our given name for $x$ as first input, and $w$ as prefix of the second input. Since $\delta_\X$ is effectively open, we can then compute $U_w = \delta_{\X}[v_w\Baire] \in \mathcal{O}(\X)$ and $V_w = \delta_\X[w\Baire] \in \mathcal{O}(\X)$. The construction guarantees that $x \in U_w$, $U_w \cap V_w = \emptyset$ and $\bigcup_{w \in W} V_w = \X \setminus \{x\}$, where $W$ is the set of all prefixes we discover.

Since $A \subseteq \X \setminus \{x\}$, computable compactness lets us find a finite $W' \subseteq W$ with $A \subseteq \bigcup_{w \in W'} V_w$. Now $U = \bigcap_{w \in W'} U_w$ and $V = \bigcup_{w \in W'} V_w$ are computable from the given data, and satisfy our criteria.
\end{proof}
\end{theorem}

\begin{corollary}
Let $\mathbf{X}$ admit an effectively open representation, be computably Hausdorff and computably compact. Then $\mathbf{X}$ is computably metrizable.
\end{corollary}

Many results of this paper
(e.g., Theorem \ref{theo:dyadicexists}, \ref{theo:hereditary}, \ref{theo:minimalfaithful}) are about computably compact computably metrizable spaces (CCCMS). This corollary shows that they also hold for computably Hausdorff computably compact spaces with effectively open representations.

\subsection{Introducing $\T^\omega$-represented spaces}
\label{section-introtomega}
We can consider $\T$ as a represented space (over $\Cantor$) via the representation $\delta_{\T} : \Cantor \to \T$ defined by $\delta_{\T}(0^\omega) = \bot$, $\delta_{\T}(p) = 0$ iff $\min \{n \in \mathbb{N} \mid p(n) = 1\}$ is even and $\delta_{\T}(p) = 1$ iff $\min \{n \in \mathbb{N} \mid p(n) = 1\}$ is odd. From this representation we derive a representation $\delta_{\T^\omega}$ of $\T^\omega$ in the usual way.
Thus we have a notion of computability of (multivalued) functions on $\T^\omega$ available. We could alternatively define computability on $\T^\omega$ directly via IM2 machines, but will not do so here for sake of simplicity.

For an alternative equivalent approach, note that there is an embedding from $\T$ to $\mathbb{S} \times \mathbb{S}$ that maps
  0 and 1 to $(\top, \bot)$ and  $(\bot, \top)$, respectively. Therefore
  $\T^\omega$ embeds into $\mathbb{S}^\omega$ which has a natural enumeration-based representation.  Thus, we have a representation $\delta_{\T^\omega}'$ of
$\T^\omega$ by restricting the standard representation of $\mathbb{S}^\omega$. The representations $\delta_{\T^\omega}$ and $\delta_{\T^\omega}'$ are equivalent.

A $\T^\omega$-representation $\psi$ of some set $X$ is just a partial surjection $\psi : \subseteq \T^\omega \to X$, and a $\T^\omega$-represented space is a set equipped with a $\T^\omega$-representation of it. As $\Cantor \subset \T^\omega$, we can consider every (ordinary) representation as a special case of a $\T^\omega$-representation. Conversely, every $\T^\omega$-representation $\psi$ induces an ordinary representation $\psi \circ \delta_{\T^\omega}$.

Let $\mathbf{X}$, $\mathbf{Y}$ be $\T^\omega$-represented spaces. We call a multivalued function $F : \subseteq \T^\omega \mto \T^\omega$ a $\T^\omega$-realizer of $f : \subseteq \mathbf{X} \mto \mathbf{Y}$ iff $\emptyset \neq \delta_Y(F(p)) \subseteq f(\delta_X(p))$ for all $p \in \dom(f \circ \delta_X)$. Unlike the situation for ordinary representations, we also need multivalued realizers here. The reason is that not every computable multivalued function $F : \subseteq \T^\omega \mto \T^\omega$ has a computable choice function\footnote{For example, consider $G : \T^\omega \mto \T^\omega$ defined by $G(p) = \{0^\omega\}$ iff $p(0) = \bot$ and $G(p) = \Cantor \setminus \{0^\omega\}$ iff $p(0) \neq \bot$.}. Again, we call a (multivalued) function between $\T^\omega$-represented spaces computable (continuous), iff it has a computable (continuous) $\T^\omega$-realizer. The following is then straight-forward:

\begin{proposition}
\begin{enumerate}
\item Let $\mathbf{X}$ and $\mathbf{Y}$ be represented spaces. A multivalued function $f : \subseteq \mathbf{X} \mto \mathbf{Y}$ is computable (continuous) as a function between represented spaces iff it is computable (continuous) as a function between $\T^\omega$-represented spaces.
\item Let $\mathbf{X}$ and $\mathbf{Y}$ be $\T^\omega$-represented spaces, and $\overline{\mathbf{X}}$ and $\overline{\mathbf{Y}}$ the induced represented spaces. Then $f : \subseteq \mathbf{X} \mto \mathbf{Y}$ is computable (continuous) iff $f : \subseteq \overline{\mathbf{X}} \mto \overline{\mathbf{Y}}$ is.
\end{enumerate}
\end{proposition}

 We thus see that the category of represented spaces and  computable (continuous) (multi-valued) functions is equivalent to the category of $\T^\omega$-represented space and computable (continuous) (multi-valued) functions. We can identify a $\T^\omega$-represented space with  the induced represented space.

\section{$\T^\omega$-representations of pruned trees}
\label{sec:cantor}
We recall that a (binary) tree is a set $T \subseteq \2^*$ such that $v \in T \ \wedge \ w \sqsubseteq v$ implies $w \in T$. The elements of a tree are called vertices. By using some standard bijection $\nu : \mathbb{N} \to \2^*$, we can then represent a tree $T$ by its characteristic function $\chi_T \in \Cantor$. We shall denote the represented space of binary trees by $\tree$.

Some $p \in \Cantor$ is called an infinite path trough a tree $T$ if $\forall n \in \mathbb{N} \ p_{\leq n} \in T$. The set of infinite paths through $T$ is denoted by $[T]$. It is well-known that the closed subsets of $\Cantor$ arise as $[T]$ for some tree in a uniform way. In other words $T \mapsto [T] : \tree \to \mathcal{A}(\Cantor)$ is computable and has a computable multivalued inverse.

A tree is pruned, if $w \in T$ implies $\exists v \in T \ w \sqsubset v$. By induction, in a pruned tree every vertex is the prefix of some infinite path through it. Moreover, for any tree $T$ there is a unique pruned tree $T_{p}$ such that $[T] = [T_p]$ -- however, $T_p$ is not computable from $T$, with the Kleene tree being the canonic counterexample.

\begin{definition}
\label{def:pruned}
We define the $\T^\omega$-represented space $\pruned$ by letting the underlying set be the pruned binary trees, and the $\T^\omega$ representation $\delta_{\pruned} : \T^\omega \to \pruned$ be defined as follows:
\begin{align*} \delta_\pruned (p) = T {
\ \stackrel{\mathrm{Def}}{\Longleftrightarrow}}
\  &\left ( p(0) = \bot \Leftrightarrow T \neq \emptyset \right ) \wedge\\
&\ \forall n \in \N
[ \left ( \nu(n) \in T \wedge p(n+1) \neq 1 \rightarrow \nu(n)0 \in T \right )
\wedge \\ & \hspace*{1.25cm} \left ( \nu(n) \in T \wedge p(n+1) \neq 0 \rightarrow \nu(n)1 \in T \right ) ]
 \end{align*}
\end{definition}

This means that we use the first symbol in a $\delta_\pruned$-name for a tree to indicate whether the tree is empty or not, with $\bot$ representing non-emptyness. If the tree is non-empty, then clearly $\varepsilon \in T$. Then for any vertex $w$ of the tree, the value of $p(\nu^{-1}(w) +1)$ indicates whether the left child, the right child or both are part of the tree. This represents precisely the pruned trees, as there the fourth case of \emph{neither} does not apply.

\begin{theorem}
\label{theo:prunedtrees}
The map $\textrm{Prune} : \tree \to \pruned$ is computable and has a computable multivalued inverse.
\begin{proof}
We compute the pruned tree $T' \in \pruned$ from $T \in \tree$, we read through $T$ layer by layer (i.e.~consider all $w \in T$ with $|w| = n$ at the same time).
 While doing so, we construct a $\delta_\pruned$-name $q$ of $T'$. We can assume that initially, $q = \bot^\omega$, and then change entries in $q$ to $0$ or $1$ as required.

If at any stage of the computation we find that $q(\nu(w)+1) = \bot$ and there is some $k \in \mathbb{N}$ such that for all $v \sqsupseteq w0$ with $|v| = k$ we learn that $v \notin T$, then we set $q(\nu^{-1}(w)+1) := 1$. If $q(\nu^{-1}(w)+1) = \bot$ and there is some $k \in \mathbb{N}$ such that for all $v \sqsupseteq w1$ with $|v| = k$ we learn that $v \notin T$, then we set $q(\nu^{-1}(w)+1) := 0$. If we ever find some $k \in \mathbb{N}$ such that $v \notin T$ for all $v$ with $|v| = k$, then we set $q(0) = 0$. Using compactness, it is straight-forward to verify that this yields a valid $\delta_\pruned$-name for $T'$. Note that moreover, every entry in the resulting name $q$ which is not specified by the definition of $\delta_\pruned$ will be either $0$ or $1$, but never $\bot$.

Now let us consider how to compute the multivalued inverse of $\textrm{Prune}$. We start with some $\delta_{\T^\omega}$-name $p \in \2^\omega$
of a $\delta_\pruned$-name $q \in \T^\omega$
of some pruned tree $T$. For $w \in \2^*$ and $n \in \mathbb{N}$, let $b_{w,n} := \delta_\pruned(\delta_{\T^\omega}(p_{\leq n}0^\omega))(\nu^{-1}(w)+1)$. Let $a_n := \delta_\pruned(\delta_{\T^\omega}(p_{\leq n}0^\omega))(0)$. Note that given $p$, $n$, $w$ we can compute $b_{w,n} \in \mathbf{3}$ and $a_n \in \mathbf{3}$.

We define a tree $T' \in \tree$ by setting $v \in T'$ iff $a_{|v|} = \bot \wedge \forall k < |v| \ ( b_{v_{\leq k}, |v|} = \bot \vee b_{v_{\leq k}, |v|} = v(k+1) )$. This is a tree by monotonicity of the condition (which is in part derived from the monotonicity of $a_n$ and $b_{w,n}$ in $n$). As $b_{w,n} \in \mathbf{3}$ is available, the condition is decidable, and thus the tree is known as an element of $\tree$. It is straight-forward to verify that $[T'] = [T]$.
\end{proof}
\end{theorem}

Let $\delta'_\pruned$ be the restriction of $\delta_\pruned$ so that
$p(n+1) \ne \bot $ if $\nu(n) \not \in \delta'_\pruned(p)$.  That is,
all non-specified entries are 0 or 1 in a name of a pruned tree.
One can see from the proof that $\textrm{Prune}$ is a map to this restricted represented space.

\begin{corollary}
\label{corr:nicefunctionscantor}
There is a surjection $t_{\Cantor}: \T^\omega \to {\A}(\Cantor)$ and a multivalued
map $s_{\Cantor}: {\A}(\Cantor) \rightrightarrows \T^\omega$ such that
\begin{enumerate}
\item $s_{\Cantor}$ and $t_{\Cantor}$ are computable.
\item $t_{\Cantor} \circ s_{\Cantor} = \id_{\A(\Cantor)}$.
\item $t_{\Cantor}$ is hereditary.
\item The restriction of $t_\Cantor$ to the domain of $\delta'_\pruned$ is faithful. 
In particular, if $A \in {\A}(\Cantor)$ is a finite set, then
the cardinality of $A$ is equal to the number of $\bot$ in $s_{\Cantor}(A)$.
\end{enumerate}
\begin{proof}
We obtain $t_{\Cantor}$ as $(T \mapsto [T]) \circ (\textrm{Prune}^{-1}) \circ \delta_\pruned$, and then $s_{\Cantor}$ as its multivalued inverse
computed through the realizer of $\textrm{Prune}$.
That $t_{\Cantor}$ is hereditary directly follows from its construction.
Property $(4)$ follows from the observation that  $s_{\Cantor}$ is a map to
the domain of $\delta'_\pruned$.
\end{proof}
\end{corollary}

We define $\psi_\Cantor$ as the faithful $\T^\omega$-representation
obtained by restricting  $t_\Cantor$ to the domain of $\delta'_\pruned$.

\section{Proper computable dyadic subbases}
\label{sec:propersubbases}

In order to obtain a result akin to  Corollary \ref{corr:nicefunctionscantor} for a larger class of spaces, we will utilize the notion of a proper dyadic subbase.  
It was introduced in \cite{tsuiki},
and further studied in \cite{tsuiki2,tsuiki4,TsukamotoTsuiki:2016,tsukamotophd}.
The original motivation was to generalize the role of the binary and signed binary representations of real number: A proper dyadic subbase induces both (1) a ``tiling'' coding generalizing the binary expansion and (2) ``covering'' coding (which forms an admissible representation) generalizing the signed binary expansion.

The definition of a (not necessarily proper) dyadic subbase was changed in \cite{TsukamotoTsuiki:2016} compared to the previous literature.  Here, we adopt the definition from \cite{TsukamotoTsuiki:2016} and introduce the notion of a computable dyadic subbase of a represented space.

\subsection{Computable dyadic subbases}

\begin{definition}\rm\label{def:dyadicsubbase}
A \emph{dyadic subbase} over a set $X$ is
a map $S : {\mathbb N} \times \2 \to \mathcal{P}(X)$ such that $S(n,0) \cap S(n,1) = \emptyset$ for every $n \in \mathbb N$ and if $\{(n,i) \mid x \in S(n,i)\} = \{(n,i) \mid y \in S(n,i)\}$ for $x, y \in X$, then $x = y$.
\end{definition}

We write $S_{n,i}$ for $S(n,i)$ and $S_{n,{\bot}} = X \setminus (S_{n,0} \cup S_{n,1})$.
For a dyadic subbase $S$ and $p \in \T^{\omega}$, define
\begin{align}
&S(p)=\bigcap_{k \in \dom(p)}S_{k,p(k)}\, .
\end{align}
We say that $S$ is a dyadic subbase of a topological space $X$ if $\{S(n, i) \mid n \in \N, i \in \2\}$ is a subbase of $X$, or equivalently,
$\{S(e) \mid e \in \T^*\}$ is a base of $X$.

A dyadic subbase $S$ defines an injection
$\varphi_S$ from $X$ to $\T^\omega$ as follows.
$$
\varphi_{S}(x)(n) = \left\{\begin{array}{ll}
0 &( x \in S_{n,0}),\\
1 &( x \in S_{n,1}),\\
\bot &(x \in S_{n,\bot}).
\end{array}
\right.
$$
Any dyadic subbase $S$ over a set $X$ thus induces a $\T^\omega$-representation $\varphi^{-1}_S : \subseteq \T^\omega \to X$. We shall denote the resulting ($\T^\omega$-)represented space by $\X_S$.
On the other hand, any injective $\T^\omega$-representation $\delta$ induces a dyadic subbase $S$
  defined as $S(n, i) = \{x \mid \delta^{-1}(x)(n) = i \}$. Therefore, a dyadic subbase can be identified with an injective $\T^\omega$-representation.

\newcommand{\I}{{\mathbb I}}
\begin{example}\label{ex:gray} Let $\I$ be the interval $[-1, 1]$.  Let $\mathbf t: \I \to \I$ be the tent function
${\mathbf t}(x) = 1- 2|x|$.
We define  the dyadic subbase $G: \N \times \2 \to \cal{P}(\I)$ recursively as $G(0,0) = [-1, 0)$, $G(0,1) = (0, 1]$,
$G(n, i) = {\mathbf t}^{-1}(G(n-1, i))$.  We have
\begin{align*}
\varphi_G(x)(0) &= \left \{\begin{array}{ll}0 & (x < 0)\\
                              \bot & (x = 0)\\
                               1   &  (x > 0)\end{array}\,,\right .\\
\varphi_G(x)(n) &= \varphi_G({\mathbf t}(x))(n-1)\,.
\end{align*}
We have
$G(010^n) = (-1/2^{n+1},0)$,  $G(110^n) = (0, 1/2^{n+1})$,
$G(\bot 10^n) = (-1/2^{n+1}, 1/2^{n+1})$, and $\varphi_G(0) = \bot 1 0^{\omega}$.
$\varphi_G(x)$ contains $\bot$ if $x$ is a dyadic rational (numbers of the form $m/2^n)$, and the sequence after a $\bot$ is always $10^\omega$.
One can easily see that $G$ is a dyadic subbase of $\I$ with the Euclidean topology.

\end{example}

\begin{definition}\label{def:dyadicsubbase3}
We say that $S$ is a (computable) dyadic subbase of a represented space $\mathbf{X}$ if
$\X_S$ is computably isomorphic to $\mathbf{X}$.
\end{definition}

The adjective \emph{computable} is redundant, but in order to avoid confusion between a dyadic subbase of a topological space and of a represented space, we add this adjective to the latter case.
Recall that a represented space $\X$ is an admissible represented space if
the map $x \mapsto \{U \mid x \in U\} : \X \to \mathcal{O}(\mathcal{O}(\mathbf{X}))$ has a computable partial inverse.
This definition also applies to $\T^\omega$-represented spaces.
The following proposition says that only admissible represented spaces have computable dyadic subbases.

\begin{proposition}
  $\X_S$ is an admissible represented space.
\end{proposition}
\begin{proof}
  With the representation $\varphi^{-1}_S \circ \delta'_{\T^\omega}$,
$p \in \2^\omega$ is a name of $x$ if $p$ is an enumeration of $\{(n,i) \mid x \in S(n,i)\}$.
That is, it is the standard representation with respect to the subbase
$\{S(n,i) \mid n \in \N, i \in \{0,1\} \}$, which is admissible.
\end{proof}

We can characterize a computable dyadic subbase as follows.
Recall that $\mathbf{3} = (\{0,1,\bot\}, \delta_\mathbf{3})$ and
$\mathcal{O}(\mathbf{3}^*)$ can be seen as the space of enumerations of sets of finite words over $\{0,1,\bot\}$. Therefore, $\{\bigcup_{e \in I} S(e) \mid  I \in \mathcal{O}(\mathbf{3}^*)\}$ is a base of $X$.

\begin{proposition} \label{prop:computablepds}
Suppose that $\mathbf{X}$ is an admissible represented space.
A dyadic subbase $S$
is a computable dyadic subbase of $\mathbf{X}$ iff $I \mapsto \bigcup_{e \in I} S(e): \mathcal{O}(\mathbf{3}^*) \to \mathcal{O}(\mathbf{X})$ is computable and has a computable multi-valued inverse $D : \mathcal{O}(\mathbf{X}) \mto \mathcal{O}(\mathbf{3}^*)$.

\begin{proof}
Suppose that the right hand side of the statement holds. Then in particular, $S : \mathbb{N} \times \mathbf{2} \to \mathcal{O}(\mathbf{X})$ is computable. As $(x,U) \mapsto (x \in U?) : \mathbf{X} \times \mathcal{O}(\X) \to \mathbb{S}$ is by definition computable, and the promise that $S(n,0) \cap S(n,1) = \emptyset$, we find that $\chi : \mathbf{X} \times \mathbb{N} \to \T$ mapping $(x,n)$ to $b$ iff $x \in S(n,b)$ is computable. This lets us compute $\id : \X \to \X_S$.

To show that $\id : \X_S \to \X$ is computable, too, recall that for admissible $\X$ the map $x \mapsto \{U \mid x \in U\} : \X \to \mathcal{O}(\mathcal{O}(\mathbf{X}))$ has a computable partial inverse. Thus, it suffices to show that $(x,U) \mapsto (x \in U?) : \X_S \times \mathcal{O}(\mathbf{X}) \to \mathbb{S}$ is computable.
Since $D$ is computable, it suffices to show the
computability of the map $\eta : \X_S \times \mathcal{O}(\mathbf{3}^*) \to \mathbb{S}$ mapping $(x,I)$ to $\top$ iff $\exists e \in I \ x \in S(e)$. As $\mathbf{3}^*$ is computably overt (which makes existential quantification computable), this in turn follows from the computability of $(x,e) \mapsto (x \in S(e)?) : \X_S \times \mathbf{3}^* \to \mathbb{S}$.

The reverse implication is obvious from how $\X_S$ inherits its topology as a subspace of $\T^\omega$.
\end{proof}
\end{proposition}

For a dyadic subbase $S$ and $p \in \T^{\omega}$, define
\begin{align}
&\bar{S}(p)= \bigcap_{k \in \dom(p)} (X \setminus {S_{k,1-p(k)}}) =
\bigcap_{k \in \dom(p)} ({S_{k,p(k)}} \cup {S_{k,{\bot}}}).
\end{align}

We have
\begin{alignat}{2}
x\in S(p)&\  \  \Leftrightarrow\  \  \varphi_S(x)(k)=p(k)\text{ for }k\in \dom (p) &
\  \  \Leftrightarrow\  \   &\varphi_S(x) \succeq p
,\label{eq:s}\\
x\in \bar{S}(p) &\  \  \Leftrightarrow\  \
\varphi_S(x)(k) \preceq p(k)\text{ for } k\in \dom(p)&
\  \  \Leftrightarrow\  \   &\varphi_S(x)\comp p.\label{eq:cs}
\end{alignat}
These equations show that $S$ and $\bar{S}$ are order-theoretic notion in $\T^\omega$.

\begin{proposition}
\label{prop:dyadicsubbase}
Let $S$ be a computable dyadic subbase of $\mathbf{X}$. Then $e \mapsto S(e) : \mathbf{3}^* \to \mathcal{O}(\mathbf{X})$ and $p \mapsto \overline{S}(p) : \T^\omega \to \mathcal{A}(\mathbf{X})$ are computable.
\begin{proof}
This follows immediately from the definition of $S(e)$ and $\overline{S}(p)$ together with the observation that finite intersection $\cap : (\mathcal{O}(\mathbf{X})^* \to \mathcal{O}(\mathbf{X})$, countable intersection $\bigcap : \mathcal{A}(\mathbf{X})^\omega \to \mathcal{A}(\mathbf{X})$ and complement $\phantom{U}^C : \mathcal{O}(\mathbf{X}) \to \mathcal{A}(\mathbf{X})$ are computable (eg \cite{pauly-synthetic}).
\end{proof}
\end{proposition}

\begin{corollary}
\label{corr:compactsemicharac}
Let $S$ be a computable dyadic subbase of computably compact $\X$. Then $\{e \in \mathbf{3}^* \mid \overline{S}(e) = \emptyset\}$ is recursively enumerable.
\end{corollary}

\begin{proposition}
\label{prop:overtcharac}
Let $S$ be a computable dyadic subbase of $\mathbf{X}$. Then $\mathbf{X}$ is computably overt iff $\{e \in \mathbf{3}^* \mid S(e) \neq \emptyset\}$ is recursively enumerable.
\begin{proof}
By definition, $\X$ is computably overt iff $\operatorname{isNonempty} : \mathcal{O}(\X) \to \mathbb{S}$ is computable. By composing this with computable $(e \mapsto S(e)) : \mathbf{3}^* \to \mathcal{O}(\X)$, we conclude that $\{e \in \mathbf{3}^* \mid S(e) \neq \emptyset\}$ is recursively enumerable.

For the converse direction, note that from $U \in \mathcal{O}(\X)$ we can compute $I \in \mathcal{O}(\mathbf{3}^*)$ with $U = \bigcup_{e \in I} S(e)$, and that $U \neq \emptyset$ iff $\exists e \in I \ S(e) \neq \emptyset$.
\end{proof}
\end{proposition}

\subsection{Proper dyadic subbases}

\begin{definition}
We say that a dyadic subbase $S$ is {\em proper} if
$\textrm{cl}\,S(e)=\bar{S}(e)$ for every $e \in \T^*$.
\end{definition}

If $S$ is a proper dyadic subbase,
$S_{n,0}$ and  $S_{n,1}$ are regular open sets which are exteriors of each other.  That is, $S_{n,\bot}$ is the common boundary between them and
$\textrm{cl}\, S_{n,i} = S_{n,i} \cup S_{n,\bot}$.
Therefore, 
if $S$ is a proper dyadic subbase then
a sequence $\varphi_S(x)$
not only contains information
$\{S_{n,i} \mid \varphi_S(x)(n) = i\}$ on
the basic open sets $x$ belongs to, but also information
$\{S_{n,i} \mid \varphi_S(x)(n) \in \{i,\bot\}\}$
on the basic open sets to whose closure $x$ belongs.

\begin{proposition}
\label{prop:propercharac}
Let $S$ be a dyadic subbase.   The followings are equivalent:
\begin{itemize}
\item[(1)] $S$ is a proper dyadic subbase.  That is, $\textrm{cl}\,S(e)=\bar{S}(e)$ for every $e \in \T^*$.
\item[(2)] $\textrm{cl}\,S(e)=\bar{S}(e)$ for every $e \in \{0,1\}^*$.
\item[(3)] $S(e) = \emptyset \Leftrightarrow \overline{S}(e) = \emptyset$
for every $e \in \T^*$.
\item[(4)] $S(e) = \emptyset \Leftrightarrow \overline{S}(e) = \emptyset$
for every $e \in \{0,1\}^*$.
\end{itemize}
\begin{proof} The equivalence of (1) and (2) is given in Lemma 9 of \cite{tsuiki2}.
The equivalence of (1) and (3) is given in Proposition 2.7 of \cite{tsuiki3}.
The equivalence of (1) and (4) is similar.
\end{proof}
\end{proposition}

\begin{corollary}\label{cor:proper}
  If $S(e) \ne \emptyset$ for every $e \in \{0,1\}^*$, then $S$ is a proper dyadic subbase.
\end{corollary}

\begin{example}
  The Gray subbase is proper by Corollary \ref{cor:proper}.
\end{example}

\begin{proposition}\label{prop:pi02}
If $S$ is a computable dyadic subbase of a computably compact and computably overt space $\X$, then $S$ being proper is a $\Pi^0_2$-property.
\end{proposition}
\begin{proof}
  From
  Corollary \ref{corr:compactsemicharac}, Proposition \ref{prop:overtcharac},
  and Proposition \ref{prop:propercharac}.
\end{proof}

\begin{proposition}\label{prop:decidable}
If $S$ is a proper computable dyadic subbase of a computably compact computably overt space $\X$, then
it is decidable whether $S(e) = \emptyset$ for $e \in \mathbf{3}^*$.
\end{proposition}
\begin{proof}
These observations follow from Proposition \ref{prop:propercharac} in combination with Corollary \ref{corr:compactsemicharac} and Proposition \ref{prop:overtcharac}.
\end{proof}

\begin{corollary}
\label{corr:dyadicsubbase}
Let $S$ be a proper computable dyadic subbase of computably overt $\mathbf{X}$. Then $e \mapsto \overline{S}(e) : \mathbf{3}^* \to \left ( \mathcal{A}(\mathbf{X}) \wedge \mathcal{V}(\mathbf{X}) \right )$ is computable.
\begin{proof}
As $w \mapsto w\bot^\omega : \mathbf{3}^* \to \T^\omega$ is computable, we obtain $e\mapsto \overline{S}(e) : \mathbf{3}^*
 \to \mathcal{A}(\mathbf{X})$ from Proposition \ref{prop:dyadicsubbase}. By definition of a proper dyadic subbase, $\overline{S}(e) = \textrm{cl}\, S(e)$ for $e \in \T^*$, and $\textrm{cl} : \mathcal{O}(\mathbf{X}) \to \mathcal{V}(\mathbf{X})$ is computable for computably overt $\X$ (eg \cite{pauly-synthetic}). We thus obtain $e \mapsto \overline{S}(e) : \mathbf{3}^* \to \mathcal{V}(\mathbf{X})$.
\end{proof}
\end{corollary}

\begin{definition}
\label{def:excase}
For a dyadic subbase $S$ of a space $\X$,
$p \in \T^{\omega}$ and $n \in \N$,
we define $\exS^n(p) \subseteq X$
and $\exbarS^n(p) \subseteq X$ as follows.
\begin{alignat*}{2}
\exS^n(p) &= \bigcap_{k < n}S_{k,p(k)}, \\
\exbarS^n(p) &= \bigcap_{k < n}\cl\, S_{k,p(k)}.
\end{alignat*}
Note that
$\cl\, S_{n,\bot} = S_{n,\bot}$.
For $e \in \T^*$, we define $\exS(e) = \exS^{|e|}(e)$ and  $\exbarS(e) = \exbarS^{|e|}(e)$.  We have
\begin{alignat}{2}
x\in \exS(e)&\  \  \Leftrightarrow\  \   \varphi_S(x)_{\leq |e|} = e,& \label{eq:bars}\\
x\in \exbarS(e) &\  \  \Leftrightarrow\  \   \varphi_S(x)_{\leq |e|} \preceq  e.&\label{eq:barcs}
\end{alignat}
\end{definition}

The sets $\exS(e)$ may fail be to be either open or closed, but are guaranteed to be $\boldd^0_2$. While each $\exbarS(e)$ is a closed set, even for a computable dyadic subbase this does not yield an effective statement in general: The sets $\cl\, S_{k,0}$, $\cl\, S_{k,1}$ will be available as elements of $\mathcal{V}(\mathbf{X})$, whereas $S_{n,\bot}$ is available as an element of $\mathcal{A}(\mathbf{X})$
but not as an element of $\mathcal{V}(\mathbf{X})$ in general because finite intersection is not a continuous operation on $\mathcal{V}(\mathbf{X})$. However, if $S$ is proper, then $\cl\, S_{k,0} = S_{k,1}^C$ and $\cl\, S_{k,1} = S_{k,0}^C$ -- i.e.~all component sets of $\exbarS(e)$ are available as elements of $\mathcal{A}(\X)$, and this space is effectively closed under intersection. We thus find:

\begin{observation}
\label{obs:effectivelyproper}
Let $S$ be a proper computable dyadic subbase of $\X$. Then $\exbarS : \mathbf{3}^* \to \mathcal{A}(\X)$ is computable.
\end{observation}

We can use this to characterize compactness similar to the characterization of overtness in Proposition \ref{prop:overtcharac} (and extending Corollary \ref{corr:compactsemicharac}):

\begin{proposition}
\label{prop:compactcharac}
Let $\mathbf{X}$ be compact and admit a proper computable dyadic subbase $S$.
Then $\X$ is computably compact iff $\{e \in \mathbf{3}^* \mid \exbarS(e) = \emptyset \}$ is recursively enumerable.
\begin{proof}
The forward-implication follows from Observation \ref{obs:effectivelyproper} and the basic characterization of computable compactness (see \cite{pauly-synthetic}). For the other direction, we need to show that we can semidecide $U = \X ?$ given $U \in \mathcal{O}(\mathbf{X})$. By Proposition \ref{prop:computablepds}, we can assume $U$ to be given as $U = \bigcup_{e \in I} S(e)$ for some enumerated set $I \subseteq \mathbf{3}^*$. Since we assume $\X$ to be compact, if $U = \X$, then already $\X = S(e_0) \cup \ldots \cup S(e_k)$ for some finite $\{e_i \mid i \leq k\} \subseteq I$.

Now from any $e_0, \ldots, e_k \in \mathbf{3}^*$ we can compute $d_0,\ldots,d_{l} \in \mathbf{3}^*$ such that $\X \setminus \left (S(e_0) \cup \ldots \cup S(e_k) \right ) = \exbarS(d_0) \cup \ldots \cup \exbarS(d_l)$. Thus, by semideciding whether $\exbarS(d_j) = \emptyset$ for all $j \leq l$, we can semidecide whether  $\X = S(e_0) \cup \ldots \cup S(e_k)$.
\end{proof}
\end{proposition}

\subsection{Existence of proper computable dyadic subbases}
\label{sec:existenceproper}

In \cite[Section 3.1]{tsukamoto}, \name{Tsukamoto} had shown that any separable metric space admits a proper dyadic subbase.
In this subsection, we effectivize the proof to show that every computably compact computable metric space admits a proper computable dyadic subbase.

We recall the definition of a computable metric space:

\begin{definition}
A computable metric space $(X,d,\alpha)$ is a separable metric space $(X,d)$ with metric $d : X \times X \to \mathbb{R}$ and a dense sequence $\alpha : \N \to X$ such that $d \circ (\alpha \times \alpha ) : \N^2 \to \Real$ is a computable double sequence of real numbers.\\

A computable metric space is turned into a represented space via the Cauchy representation $\delta_X :\subseteq \N^\N \to X$ defined by $\delta_X ( p) := x$ iff $\forall n \ d(\alpha(p(n)),x) < 2^{-n}$. A computable Polish space is a represented space induced by a complete computable metric space.
\end{definition}

We point out that since a computable metric space has by definition a computable dense sequence, it is always computably overt. In particular, $\cl : \mathcal{O}(\X) \to \mathcal{V}(\X)$ is computable for any computable metric space $\X$.

Next, we gather some results and definitions which we will require for the existence proof.

\begin{theorem}(Computable Baire category theorem \cite{brattka7})
Let $\X$ be a computable Polish space.
There exists a computable operation $\Delta :\subseteq \A(\X)^\N \times \O(\X) \mto  \X$ such that,  for any sequence $(A_n)_{n \in \N}$ of closed nowhere dense subsets of $X$ and a non-empty open subset $I$,
$\Delta((A_n)_{n \in \N}, I) \subseteq I \setminus \cup_{n=0}^{\infty} A_n$.
\end{theorem}

\begin{definition}
  $c \in \mathbb{R}$ is a local maximum of a continuous function $f : \mathbf{X} \to \mathbb{R}$ if $c$ is the maximum value of $f|_V$ for some open subset $V$.  Local maximum and local minimum values are called local extrema.
We denote by $\extr(f)$ the set of local extrema of $f$.
\end{definition}

For a function $f: X \to \Real$ and a real number $c$,  let
$U^0(f, c) = \{x \in X : f(x) < c\}$ and
$U^1(f, c) = \{x \in X : f(x) > c\}$. Note that $U^b : \mathcal{C}(\X,\Real) \times \Real \to \mathcal{O}(\X)$ is a computable function. While $U^0(f,c)$ and $U^1(f,c)$ are clearly always disjoint, in general they do not have to be exteriors of each other. However, we can establish:

\begin{lemma}\label{lemma0}
If $c \notin \extr(f)$, then $U^0(f, c)$ and $U^1(f, c)$  are exteriors of each other.
\end{lemma}

\begin{proof}
  Suppose that $c$ is not a local extremum of $f$ and $f(x) = c$.
  Then, for any $V \ni x$, $V \cap U^1(f, c) \ne \emptyset$ and
  $V \cap U^0(f, c) \ne \emptyset$.
\end{proof}

Let $B(a, r)$ and $\bar{B}(a,r)$ denote the open and respectively closed ball of  center $a$ and radius $r$. For $A \subseteq \X$, $f : \X \to \mathbb{R}$, $a \in \X$ and $r \in \mathbb{R}$, let $\bar{M}_{A,f}(a,r)$ be the maximum value of $f$ in $\bar{B}(a, r) \cap A$ and $M_{A,f}(a,r)$ be the supremum value of $f$ in $B(a, r) \cap A$. Let $D_{A,f}(a, r) = \{x : M_{A,f}(a, 2r) \leq x \leq \bar{M}_{A,f}(a, r) \}$.

\begin{lemma}
\label{lem:mmd}
The following functions are computable:
\begin{enumerate}
\item $(A,f,a,r) \mapsto \bar{M}_{A,f}(a,r) : \mathcal{K}(\X) \times \mathcal{C}(\X,\mathbb{R}) \times \X \times \mathbb{R} \to \mathbb{R}_>$,
\item $(A,f,a,r) \mapsto M_{A,f}(a,r) : \mathcal{V}(\X) \times \mathcal{C}(\X,\mathbb{R}) \times \X \times \mathbb{R} \to \mathbb{R}_<$,
\item $(A,f,a,r) \mapsto D_{A,f}(a, r) : (\mathcal{V}(\X) \wedge \mathcal{K}(\X))  \times \mathcal{C}(\X,\mathbb{R}) \times \X \times \mathbb{R} \to \mathcal{A}(\mathbb{R})$.
\end{enumerate}
\begin{proof}
\begin{enumerate}
\item From $a$, $r$ we can compute $\bar{B}(a,r) \in \mathcal{A}(\X)$. The intersection of a compact and a closed set is computable as a compact set. The maximum of a real-valued continuous function on a compact set is approximable from above, i.e.~computable as a point in $\mathbb{R}_>$.
\item The map $(U, A) \mapsto \cl (A \cap U) : \mathcal{O}(\X) \times \mathcal{V}(\X) \to \mathcal{V}(\X)$ is computable. We can thus compute $\cl\,  (B(a, r) \cap A) \in \mathcal{V}(\X)$. Taking the closure does not impact the supremum of a continuous function. The supremum of a real-valued continuous function on an overt set is approximable from below, i.e.~computable as a point in $\mathbb{R}_<$.
\item By combining $(1)$ and $(2)$, and noting that $x \mapsto \{y \mid y \leq x\} : \mathbb{R}_> \to \mathcal{A}(\mathbb{R})$ and $x \mapsto \{y \mid x \leq y\} : \mathbb{R}_< \to \mathcal{A}(\mathbb{R})$ are computable.
\end{enumerate}
\end{proof}
\end{lemma}

\begin{lemma}
\label{lem:extrfcharac}
$\extr(f|_{A}) = \bigcup_{n, k \in \N} D_{A,f}(\alpha(n), 2^{-k}) \cup D_{A,-f}(\alpha(n), 2^{-k})$.
\begin{proof}
Since $\bar{B}(a, r) \subseteq B(a, 2r) $, we find that $\bar{M}_{A,f}(a, r) \leq M_{A,f}(a, 2r)$. Thus, $D_{A,f}(a, r)$ is either empty, or the singleton $\{c\}$ where $c = \bar{M}_{A,f}(a, r) = M_{A,f}(a, 2r) = \max_{x \in B(a,2r) \cap A} f(x)$. In either case, we see that $D_{A,f}(a, r) \subseteq \extr(f|_{A})$. The same argument works for $D_{A,-f}(a, r)$ by exchanging minima and maxima.

Conversely, let $c \in \extr(f|_{A})$. By moving to $-f$ if necessary, we can assume $c$ to be a local maximum. This means there is some open $V$ such that $c = \max_{x \in A \cap V} f(x)$. Pick some $x_0 \in A \cap V$ with $f(x_0) = c$. Since we are working in a separable metric space, there are $n ,k \in \mathbb{N}$ with $x_0 \in B(\alpha(n),2^{-k}) \subseteq B(\alpha(n),2^{-k+1}) \subseteq V$. We then find that $c = \bar{M}_{A,f}(\alpha(n), 2^{-k}) = M_{A,f}(\alpha(n), 2^{-k+1})$, hence $D_{A,f}(\alpha(n),2^{-k}) = \{c\}$.
\end{proof}
\end{lemma}

\begin{lemma}\label{lemma1}
  Let $(X,d,\alpha)$ be a computable metric space.
There exists a computable operation $\Gamma :\subseteq \left ((\K(\X) \wedge  \V(\X)) \times
\C(\X, \Real)\right )^* \times \O(\X)  \mto \Real$ which returns a real in
$I \setminus\, \left (\bigcup_{i \leq n} \extr(f_i|_{A_i}) \right )$ to input $((A_0, f_0),\ldots,(A_n,f_n), I)$ if $I$ is nonempty.
\end{lemma}

\begin{proof}
By Lemma \ref{lem:mmd} (3) and Lemma \ref{lem:extrfcharac}, we can compute each $\extr(f_i|_{A_i})$ as the union of a sequence of nowhere-dense closed sets. The computable Baire category theorem then lets us find a point avoiding their union.
\end{proof}

\begin{theorem}
\label{theo:dyadicexists}
Every CCCMS (computably compact computable metric space) $(X, d, \alpha)$ admits a proper computable dyadic subbase.
\end{theorem}
\begin{proof}
Let $(I_n)$ be a computable sequence of  open rational intervals which forms a base of $\Real^{\geq 0}$.   Let $\langle \pi_1, \pi_2 \rangle: \N \to \N \times \N$ be a standard bijection. Let $f_n(x) = d(\alpha(\pi_1(n)), x)$.
Then, for any choice of $c(n) \in I_{\pi_2(n)}$, we obtain a dyadic subbase of $\X$ by $S(n,b) = U^b(f_n, c(n))$. If $c$ is computable, then $S$ is even a computable dyadic subbase.

By Proposition \ref{prop:propercharac}, for $S$ to be proper, we need that
$\textrm{cl}\,S(e)=\bar{S}(e)$ for every $e \in \{0,1\}^*$. We will choose $c$ inductively to ensure this.

For a sequence $(c(k))_{k < n}$ of length $n$ and $e \in \{0,1\}^n$,
we define $S(n, (c(k))_{k < n}, e) \in \O(\X)$ and $\bar{S}(n, (c(k))_{k < n}, e) \in \A(\X)$
as  $S(n, (c(k))_{k < n}, e) = \cap_{i<n} U^{e(i)}(f_i, c(i))$ and
$\bar{S}(n, (c(k))_{k < n}, e) = \cap_{i<n} X \setminus\, U^{1-e(i)}(f_i, c(i))$. Since $\X$ is computably compact, we also obtain $\bar{S}(n, (c(k))_{k < n}, e) \in \K(\X)$.

First, we choose $c(0) \in I_0$ avoiding the local extrema of $f_0$, i.e.~$c(0) \in \Gamma((\X, f_0), I_0)$ via Lemma \ref{lemma1}.

Suppose that we have defined $c(0)$ to $c(n-1)$ so that
$\bar{S}(n, (c(k))_{k < n}, e) = \cl\,  S(n, (c(k))_{k < n}, e)$ for every $e \in \{0,1\}^n$.
Let $S(e) = S(n, (c(k))_{k < n}, e)$ and $\bar{S}(e) = \bar{S}(n, (c(k))_{k < n}, e)$.  We have $\bar{S}(e) = \cl\, S(e) \in \K(\X) \wedge \V(\X)$.
We choose $c(n) \in \Gamma((\bar{S}(e), f_n))_{e \in \{0,1\}^n}), I_n)$ via Lemma \ref{lemma1}, i.e.~avoiding any local extrema of $f_n$ on any $\bar{S}(e)$.

We prove $\bar{S}(n+1, (c(k))_{k \leq n}, e') = \cl\,  S(n+1, (c(k))_{k \leq n}, e')$ for every $e' \in \{0,1\}^{n+1}$.  We show it for the case $e' = e0$.
we need to show
\[
\cl (S(e) \cap U^{0}(f_n, c(n))) = \bar{S}(e) \setminus\, U^{1}(f_n, c(n)).
\]
First, we have
\[
  \cl (S(e) \cap U^{0}(f_n, c(n))) = \cl_{\cl\, S(e)}(\cl\, S(e) \cap U^0(f_n, c(n)))
\]
because $\cl\, (A \cap B) = \cl_{\cl\,{A}} ((\cl\, A) \cap B)$ for open sets $A, B$.
Since $c(n)$ is not a local extremum of $f_n$ on $\cl\, S(e) = \bar{S}(e)$, by Lemma \ref{lemma0},
\[
\cl_{\cl\, S(e)}(\cl\, S(e) \cap U^0(f_n, c(n))) = \cl\, S(e) \setminus U^1(f_n, c(n)) = \bar{S}(e) \setminus\, U^{1}(f_n, c(n)).
\]

With the dyadic subbase $S$ thus defined, $S_{n,0} = \{x \mid d(x, \alpha(\pi_1(n))) < c(n)\}$ and
$S_{n,1} = \{x \mid d(x, \alpha(\pi_1(n))) > c(n)\}$
are computable open sets of $\X$ and therefore $F: I \mapsto \bigcup_{e \in I} S(e): \mathcal{O}(\mathbf{3}^*) \to \mathcal{O}(\mathbf{X})$ is computable.
On the other hand,  for a base $B(\alpha(n_1), r)$ for $n_1 \in \N$ and $r \in \Q^{>0}$, take an $n_2$ such that $I_{n_2} \subset [0, r]$ and
define $n = \langle n_1, n_2 \rangle$. Then,
$S_{n,0} \subset B(\alpha(n_1), r)$.  Therefore, the multi-valued inverse of $F$ is computable.
Thus, by Proposition \ref{prop:computablepds},
$S$ is a computable dyadic subbase of $\X$.
\end{proof}

The proof of Theorem \ref{theo:dyadicexists} is fully uniform. We could have constructed the represented spaces of all compact separable metric spaces, similar to what is suggested in \cite{rettinger2}. Then we would have stated that we can compute a proper dyadic subbase from a given compact separable metric space.

Given that the original proof from \cite[Section 3.1]{tsukamoto} does not require compactness, of course the question arises why we require it for Theorem \ref{theo:dyadicexists}. Computable compactness of $\X$ is used to obtain the $S(e) \in \mathcal{K}(\X) \wedge \V(\X)$, not just in $\mathcal{A}(\X) \wedge \mathcal{V}(\X)$, which is needed to invoke Lemma \ref{lemma1}. We will see in the following example that without compactness, there is no computable way to avoid local extrema. There might still be a way to salvage the construction by exploiting that we do not need to avoid local extrema of arbitrary functions, but of distances.

\begin{example}
There exists a computable function $f : \Baire \to [0,1]$ such that for every computable real $x \in [0,1]$ there exists some $n \in \mathbb{N}$ with $x = \max_{p \in \Baire} f(np)$.
\begin{proof}
We can describe a computable function $f : \Baire \to [0,1]$ by computably labelling the full countably branching tree $\mathbb{N}^*$ with closed rational intervals $(I_w)$ such that $|I_w| \leq 2^{-|w|+1}$ and $w \preceq u \Rightarrow I_u \subseteq I_w$. We label the root and all first-generation children by $[0,1]$. Inside the subtree below the $n$-th first generation child, we try to compute the $n$-th computable real number $x_n$ (which might be undefined).

Once we have found the label $I_{nw}$ to some vertex, we compute the labels for the successors as follows: To obtain $I_{nwk}$, we continue the computation of $x_n$ for $k$ more steps. If this yields an approximation $q$ of $x_n$ of precision $2^{-|nwk|}$, and $[q - 2^{-|nwk|}, q + 2^{-|nwk|}] \cap I_{nw}$ is non-empty, then we set $I_{nwk} := [q - 2^{-|nwk|}, q + 2^{-|nwk|}] \cap I_{nw}$. Otherwise, we set $I_{nwk} = [\min I_{nw}, \min I_{nw}]$.

This ensures that if $p$ is an upper bound for the computation time of $x_n$, then $f(np) = x_n$. Otherwise, $f(np) < x_n$.
\end{proof}
\end{example}

\begin{corollary}
The map $\operatorname{AvoidExtrema} : \mathcal{C}(\Baire, [0,1]) \mto [0,1]$ defined by $$\operatorname{AvoidExtrema}(f) = [0,1] \setminus \{\max_{p \in \Baire} f(wp) \mid w \in \mathbb{N}^*\}$$ is not computable.
\end{corollary}

\section{Main result}
\label{sec:mainresult}
In this section, we show that for every computably compact computable metric space (CCCMS) $\mathbf{X}$ there is a faithful $\T^\omega$-representation of $\mathcal{K}(\mathbf{X})$.

\subsection{The poset $\widehat{K}_S$}

Our construction utilizes
a sub-poset $\widehat{K}_S$ of $\tau(\T^*)$, which was studied in \cite{tsuiki4}.

We say that $e \in \T^*$ is an \emph{immediate successor} of $d \in \T^*$
(denoted by $d \prec^1 e$)
if $d \prec e$ and $e$ has one more digit than $d$.
We say that a subset $E$ of $\T^*$  is \emph{tree-like} if
$\forall e \in E\  e \ne \epsilon \to \exists d \in E\ d \prec^1 e$.
For a tree-like set $E \subseteq \T^*$ and $e \in E$,
we write $\mathrm{succ}_E(e)$ for the set of immediate successors of $e$ in $E$.
A tree-like set $E$ is called \emph{finitely branching} if
$\mathrm{succ}_E(e)$ is finite for each $e \in E$.

For a tree-like subset $E$,
we call an infinite sequence $ \epsilon = e_0\prec^1 e_1 \prec^1 \ldots$ in $E$ an \emph{infinite path} in $E$.
Every infinite path $P= (e_i)_{i \in \N}$  has a limit $\sqcup_{i \in \N} \iota(e_i)$,
which is an element of $\T^\omega \setminus \T^*$.
We define $L(E) \subseteq \T^\omega$ as the set of limits of infinite paths in $E$ (i.e.~$L(E)$ generalizes the notion of body of a tree).



The following proposition gives a sufficient condition for $\bar{S}(p)$ to become a singleton.
\begin{proposition}[Proposition 2.5 of \cite{tsuiki4}]\label{prop:rho}
  Suppose that $S$ is a proper dyadic subbase of a Hausdorff space $X$.
\begin{itemize}
\item[(1)]  If $x \ne y \in X$, then $x \in S_{n, a}$ and $y \in S_{n, 1-a}$ for some $n$ and $a$.

  \item[(2)] If $x \in X$ and $p \succ \varphi_S(x)$, then $S(p) = \emptyset$.

\item[(3)] If $x \in X$ and $p \succeq \varphi_S(x)$, then $\bar{S}(p) = \{x\}$.
Therefore,  another $\T^\omega$-representation  $\rho_S$ of $X$ is derived from $S$.
The domain of $\rho_S$ is $\uparrow\!\! \varphi_S(X)$ and $\rho_S(p) = x$ if and only if $ \bar{S}(p) \ni x$,
or equivalently, $p \succeq \varphi_S(x)$.
\end{itemize}
\begin{proof}
\begin{enumerate}
\item Since $X$ is Hausdorff and $S$ is proper, there is $e \in \T^*$ such that
$x \in S(e)$ and $y \not \in cl\, S(e) = \bar{S}(e)$.  Therefore,  $e \preceq \varphi_S(x)$ and
$e \not \comp \varphi_S(y)$ by equations (\ref{eq:s}) and (\ref{eq:cs}).
Thus, $\varphi_S(x) \not \,\comp \varphi_S(y)$.  Therefore,
$\varphi_S(x)(n) = 1- \varphi_S(y)(n)$ for some $n$.

\item  It holds for dyadic subbases in general.
First, we have $S(\varphi_S(x)) = \{x\}$.  For $i \in \dom(p) \setminus \dom(\varphi_S(x)), x \not \in S_{i,p(i)}$ holds.

\item  From (1), we have $\bar{S}(\varphi_S(x)) = \{x\}$.
We have $\bar{S}(p) \subseteq \bar{S}(\varphi_S(x)) = \{x\}$ from $p \succeq \varphi_S(x)$.
On the other hand, $\bar{S}(p) \ni x$ because $p \comp \varphi_S(x)$.
\end{enumerate}
\end{proof}
\end{proposition}



\begin{definition}
Let $S$ be a dyadic subbase of a space $\X$.
 We define the poset $\widehat{K}_S \subseteq \T^*$ as
 \[
\widehat{K}_S = \{\bo{p}{\leq n}  \mid  \exists x \in X\ \varphi_S(x) \preceq p , n \in \N\}\,.
\]
\end{definition}

Every element of $\widehat{K}_S$ does not end with $\bot$ and thus belongs to $\tau(\T^*)$.  Therefore, $\widehat{K}_S$ is a poset.
It is a tree-like set by construction.  Since
$\bo{p}{\leq n} = p_{\leq m}$ for some $m$,
we can derive
\[
\widehat{K}_S = \{e \in \tau(\T^*) \mid \exbarS(e) \ne \emptyset\}
\]
from   Equation (\ref{eq:barcs}).
As is shown in Proposition 4.4 of \cite{tsuiki4},
$\widehat{K}_S$ is consistently complete (i.e., every pair of consistent elements have
their least upper bound) and $D(\widehat{K}_S) = \widehat{K}_S \cup L(\widehat{K}_S)$ forms a Scott sub-domain of $\T^\omega$.

\begin{figure}
\begin{center}
\begin{tikzpicture}{align=bottom}
\draw (0,0) node (e) {$\epsilon$};
\draw (-4,1) node (0) {$0$};
\draw (0,1) node (b1) {$\bot1$};
\draw (4,1) node (1) {$1$};

\draw (e) --(b1);
\draw (e) --(1);
\draw (e) --(0);

\draw (-2,2) node (01) {$01$};
\draw (0,2) node (b10) {$\bot 10$};
\draw (2,2) node (11) {$11$};
\draw (4,2) node (1b1) {$1\bot 1$};

\draw (0) --(01);
\draw (b1) --(01);
\draw (b1) --(b10);
\draw (b1) --(11);
\draw (1) --(11);
\draw (1) --(1b1);

\draw (1,3) node (110) {};
\draw (2,3) node (11b1) {$11 \bot 1$};
\draw (3,3) node (111) {$111$};
\draw (4,3) node (1b10) {};
\draw (5,3) node (101) {};

\draw (11) --(11b1);
\draw (11) --(111);
\draw (1b1) --(111);

\draw (1.5,4) node (a) {};
\draw (2,4) node (b) {};
\draw (2.5,4) node (c) {};
\draw (3,4) node (d) {};
\draw (3.5,4) node (e) {};

\draw (11b1) --(2.5,3.5);
\draw (111) --(2.5,3.5);
\draw (111) --(3,3.5);

\draw (2.5,5) node (b1) {};
\draw (2.75,5) node (c1) {};
\draw (3,5) node (d1) {};


\draw (2.5,3.5) --(2.5,4);
\draw (2.5,3.5) --(2.75,4);
\draw (3,3.5) --(2.75,4);

\draw (2.5,4) --(2.625,4.5);
\draw (2.75,4) --(2.625,4.5);
\draw (2.75,4) --(2.75,4.5);

\draw (2.65,4.75) node {...};

\draw (-1,3) node (010) {$010$};
\draw (0,3) node (b100) {$\bot 100$};
\draw (1,3) node (110) {$110$};

\draw (01) --(010);
\draw (b10) --(010);
\draw (b10) --(b100);
\draw (b10) --(110);
\draw (11) --(110);

\draw (010) --(-0.5,3.5);
\draw (b100) --(0   ,3.5);
\draw (b100) --(-0.5 ,3.5);
\draw (b100) --(0.5  ,3.5);
\draw (110) --(0.5,3.5);

\draw (-0.5,3.5) --(-0.25,4);
\draw (0,3.5) --(-0.25,4);
\draw (0,3.5) --(0,4);
\draw (0,3.5) --(0.25,4);
\draw (0.5,3.5) --(0.25   ,4);

\draw (-0.25,4) --(-0.125,4.5);
\draw (0,4) --(-0.125,4.5);
\draw (0,4) --(0,4.5);
\draw (0,4) --(0.125,4.5);
\draw (0.25,4) --(0.125   ,4.5);

\draw (0,4.75) node {...};
\draw (-3, 5)-- (3, 5);
\draw (0, 5.2) node (Lb){$\bot10^\omega$};
\draw (-0.5, 6) node (L0){$010^\omega$};
\draw (0.5, 6) node (L1){$110^\omega$};
\draw (Lb)--(L0);
\draw (Lb)--(L1);
\draw (2.65, 5.2) node (L11){$1^\omega$};

\end{tikzpicture}
\caption{Some part of $D(\widehat{K}_G)$. \label{fig:G}}
\end{center}
\end{figure}
\begin{example}
  For the Gray subbase $G$ in Example \ref{ex:gray}, Figure \ref{fig:G} shows
 a tree-like subset of $\widehat{K}_G$
that contains
all the paths with the limits $\bot 10^\omega$, $0 10^\omega$, $1 10^\omega$, and
$1^\omega = \varphi_G(1/3)$.
Note that $\uparrow \!\! \varphi_G(0)  = \{\bot 10^\omega, 0 10^\omega, 1 10^\omega\}$.
\end{example}

\begin{proposition}\label{prop:compact}
Suppose that $S$ is a proper dyadic subbase of a compact Hausdorff space $X$.
\begin{enumerate}
\item[(1)] $\{S(e) \mid e \in \widehat{K}_{S}\}$ is a base of $X$.

\item[(2)] $S(e) = \cup \{S(d) \mid d \in \mathrm{succ}_{\widehat{K}_{S}}(e)\}$ for $e \in \widehat{K}_{S}$.

\item[(3)] $\widehat{K}_{S}$ is finitely branching.

 \item[(4)]
$L(\widehat{K}_{S}) = \uparrow\!\! \varphi_S(X)$ and therefore
the domain of $\rho_S$ is $L(\widehat{K}_{S})$.  Moreover,
the map
$\varphi_S \circ \rho_S$ from $L(\widehat{K}_{S})$ to $\varphi_S(X)$ is continuous.
Therefore, $\varphi_S(X)$ is a retract of $L(\widehat{K}_{S})$.
It also means that
$X$ is homeomorphic to the set of minimal elements of $L(\widehat{K}_{S})$.
\end{enumerate}

\begin{proof}
\begin{enumerate}
\item[(1)] If $x \in S(d)$ for $d \in \T^*$, then $e = \bo{\varphi_S(x)}{\leq|d|}$ satisfies $e \in \widehat{K}_S$ and $x \in S(e) \subseteq S(d)$.

\item[(2)] Suppose that  $x \in S(e)$.  Let $d = \varphi_S(x)_{\leq |e|}$.  If $e \prec d$, then there is an element $e \prec^1 d' \preceq d$
such that $x \in S(d')$.  If $e = d$, then $d' = \bo{\varphi_S(x)}{\mathrm{level}(d)+1}$ satisfies $e \prec^1 d'$  and $x \in S(d')$.

\item[(3)] It is proved in Proposition 5.10 of \cite{tsuiki4} for a large class of spaces that contain
compact Hausdorff spaces.

\item[(4)]  It is Proposition 3.4 and Theorem 6.4 of \cite{tsuiki4}.
\end{enumerate}
\end{proof}
\end{proposition}

We explore the relationship between the two $\T^\omega$-representations
  $\varphi_S^{-1}$ and $\rho_S$.
As a part of Proposition \ref{prop:compact}(4),
we have seen $\varphi_S(X) \subseteq L(\widehat{K}_{S})$.
Actually,  $\varphi_S(x)$ is the limit of the path
$(\bo{\varphi_S(x)}{\leq n})_n$ and $\bo{\varphi_S(x)}{\leq n} \in  \widehat{K}_S$ for every $n \in \N$.
Thus, the inclusion map $\varphi_S(X) \subseteq L(\widehat{K}_{S})$
is the realizer of $id: (X, \varphi_S^{-1}) \to (X, \rho_S)$.   The following proposition says that
if $S$ is a proper computable dyadic subbase of a CCCMS $\X$,
then
$id: (X, \rho_S) \to (X, \varphi_S^{-1})$ is also computable and therefore
the two representations are computably isomorphic.

\begin{proposition}\label{prop:computable-retract}
  Suppose that $S$ is a proper computable dyadic subbase of a CCCMS $\X$.
  Then, the retraction $\varphi_S \circ \rho_S$ from $L(\widehat{K}_{S})$ to $\varphi_S(X)$ is  computable and therefore
$id: (X, \rho_S) \to (X, \varphi_S^{-1})$ is computable.
Thus,  $\rho_S$ is also a $\T^\omega$-representation of the represented space $\X$.

\end{proposition}
\begin{proof}
  Suppose that $p \in L(\widehat{K}_{S}) \subset \T^\omega$ is given.
  Proposition \ref{prop:rho} says that there is a unique $x$ such that $\varphi_S(x) \preceq p$
  and we want to compute the sequence   $q = \varphi_S(x)$ from $p$.
  Since $q \preceq p$,
  $q$ is obtained by replacing some digits in $p$ with $\bot$.
  Therefore,
  $$\{x\} = \bar{S}(p) = \cap_{i\in \dom(p)}  (S_{i,p(i)}\cup S_{i,\bot}).$$
  In order to obtain $q$, we need to detect all $i \in \dom(p)$ such that $x \in S_{i,p(i)}$.

  Since $X$ is regular and $S$ is proper,  if $x \in S_{i, p(i)}$, then
there is $n$ such that
  $S_{i, p(i)} \supseteq
\cl\, S(q_{\leq n}) = \bar{S}(q_{\leq n})$.
  We also have $\bar{S}(q_{\leq n}) \supseteq\bar{S}(p_{\leq n}) \ni x$.
  Therefore,  $x \in S_{i,p(i)}$ if and only if $\bar{S}(p_{\leq n}) \subseteq S_{i,p(i)}$ for some $n$.
  It is semi-decidable, because $\bar{S}(p_{\leq n})$ is in $\A(X)$, $S_{i,p(i)}$ is an open set, and
  $X$ is computably compact.
\end{proof}

\subsection{Preparation for the main theorem}

Suppose that  $S$ is a proper dyadic subbase of a space $X$.
For a finite-branching tree-like subset $E$ of $\T^*$, we define
$[E]_S$ as follows
\begin{align*}
[E]_S &= \bigcap_{n \in \mathbb{N}} \bigcup_{e \in \bo{E}{n} } \overline{S}(e).\label{eq:es2}
\end{align*}
$[E]_S$ is a closed set and we consider that $[T] \in \A(\Cantor)$ for a tree $T$
in Section \ref{sec:cantor} is a special case that $X = \Cantor$ and $S$ is the trivial proper dyadic subbase.

  \begin{proposition}\label{prop:bracket}
If  $S$ is a dyadic subbase of $X$ and $E$ is a finite-branching tree-like subset of $\T^*$, then
  $[E]_S = \cup \{\bar{S}(p) \mid p \in L(E)\}$.
\end{proposition}
  \begin{proof}
   Let $x \in [E]_S$.  We have $\forall n\ \exists e_n \in \bo{E}{n} (x \in \bar{S}(e_n))$.
  We choose such $e_n$ for each $n$ and consider the increasing sequence $(\sqcup_{i\leq n} e_i)_n$.
    Let $p$ be its limit.
    Since $\bar{S}(p) = \cap_n \bar{S}(e_n)$,
    we have $x \in \bar{S}(p)$.
    On the other hand, every $p \in L(E)$ is a limit of an infinite path $(e_n)_{n \in \N}$ in $E$,
    and $\bar{S}(p) \subseteq \bar{S}(e_n)$ for each $n \in \N$.
\end{proof}

  \begin{proposition}\label{prop:tree-like}
   Let $K$ be a finitely branching tree-like subset of $\T^*$.
   There is a strictly increasing function $f : \mathbb{N} \to \mathbb{N}$
and a partial surjective monotonic map $\beta :\subseteq  \2^* \to K$ such that
(1) $\mathrm{dom}(\beta)$ is a tree,  (2) $\beta$ maps
$\2^{f(n)} \cap \mathrm{dom}(\beta)$ to $\bo{K}{n}$, and
(3) if $\beta(w) \prec^1 e$ in $K$, then
there exists a unique $v \sqsupset w$ such that $\beta(v) = e$.
Such a $\beta$ maps a tree in $\dom(\beta)$ to a tree-like subset of $K$
and injectively maps an infinite path in $\dom(\beta)$ to an infinite path in $K$,
 \end{proposition}
 \begin{proof}
 For each level $n$, let $h(n)$ be the maximum of $\lceil\log{|\mathrm{succ}_K(e)|}\rceil$ for $e \in \bo{K}{n}$.
    We define $f$ inductively as $f(0) = 0$ and
    $f(n+1) = f(n) + h(n)$.
For each $e$,  we fix an isomorphism $h_e$ from a subset of
$\2^{h(n)}$ to $\mathrm{succ}(e)$.
For example,  we order $\mathrm{succ}(e) \subset \mathbf{3}^*$
according to the lexicographic ordering of $\mathbf{3}^*$, and use the
binary notation in $\2^{h(n)}$.
Then, we define
$\beta(w)$ for $w \in \2^{f(n)}$ inductively as
$\beta(wv) = h_{\beta(w)}(v)$ for $w \in \2^{f(n)}$ and $v \in \2^{h(n)}$, when it is defined.  We extend $\beta$ to a partial map defined on a lower-closed subset of $\2^*$
by defining $\beta(w) = \beta(w_{\leq f(n)})$ if $f(n) \leq |w| < f(n+1)$.
\end{proof}

If we apply this proposition to the case $X$ is a compact Hausdorff space and
$K = \widehat{K}_{S}$,  each tree $T \subseteq \mathrm{dom}(\beta)$
is mapped by $\beta$ to a tree-like subset of $\widehat{K}_{S}$, and thus
\textsf{}specifies a closed subset $[\beta(T)]_S$ of $X$. 
This correspondence is surjective
because, given a closed subset $A$,  $T = \bigcup_{x \in A}
\beta^{-1}(\{{\varphi_S(x)_{\leq n}}: n \in \N\})$ is a tree in $\mathrm{dom}(\beta)$
such that
$[\beta(T)]_S = A$.
Note that $[\beta(T)]_S$ is determined by $L(\beta(T))$ by Proposition \ref{prop:bracket}.
Since $\beta$ surjectively maps an infinite path in $T$ to an infinite path in $\beta(T)$,
$L(\beta(T))$ is determined by $L(T)$, and is again determined by $\textrm{Prune}(T)$.
In this way, a closed subset of $X$ is determined by a pruned tree.
On the other hand, each infinite path in $\mathrm{dom}(\beta)$ is mapped to an infinite path in
$\widehat{K}_{S}$ and therefore  specifies
 a point of $X$ by Proposition  \ref{prop:compact}(4).
Combining these with the $\T^\omega$-representation of pruned trees in Section \ref{sec:cantor},
we {\bf almost} have matching representations of $\X$ and ${\cal K}(\X)$.

The obstacle is that the function $f$ in Proposition \ref{prop:tree-like} is not computable in general. This is because the cardinality of
$\mathrm{succ}(e)_{\widehat{K}_{S}}$ is not computable, and therefore
the above representation of $X$ is not computably isomorphic to the represented space  $\X_S$.
To solve this, we consider a computable sub-poset $H \subset \T^*$
such that $\widehat{K}_{S} \subseteq \tau(H)$.

\begin{lemma}\label{lemma:H}
Let $S$ be a proper computable dyadic subbase of a CCCMS $\X$.
There is a computable monotonic function $g : \mathbb{N} \to \mathbb{N}$ such that
the set $H \subset \T^*$ defined as
$$\bo{H}{n} = \{e \in \bo{\T^*}{n} \mid |e| = g(n)\}$$
satisfies $\widehat{K}_S \subseteq \tau(H)$.
$H$ is decidable as a subspace of $\mathbf{3}^*$.
\label{lem:H}

\begin{proof}
We inductively define ${\bo{H}{n}} \subset {\bo{\T^*}{n}}$.
Let $g(0) = 0$ and $\bo{H}{0} = \{\epsilon\}$.
Suppose that ${\bo{H}{n}}$ is defined.
For each element $d \in {\bo{H}{n}}$, we perform the following procedure.

By dove-tailing, search for some $ k_d \geq |d|$ such that $\exbarS^{k_d}(d)$ is empty.
Recall that we identify $d \in \T^*$ with $d\bot^\omega$ for the purpose of defining $\exbarS^k(d)$ for $k > |d|$.
Since $\mathrm{succ}_{\widehat{K}_{S}}(d)$ is finite, let $m$ be the maximal length
of elements in $\mathrm{succ}_{\widehat{K}_{S}}(d)$.  Then,
we have $\exbarS^k(d) = \emptyset$ for $k \geq m$.
Therefore, we can find a suitable candidate for $k_d$
by Proposition \ref{prop:compactcharac}.
We set $g(n+1)$ to be $\max \{k_d \mid d \in {\bo{H}{n}}\}$ and
${\bo{H}{n+1}} = \{d \in {\bo{\T^*}{n+1}} \mid |d| = g(n+1)\}$.

We show ${\bo{(\widehat{K}_S)}{n}} \subseteq \tau(H)$ by induction on $n$.
We have ${\bo{(\widehat{K}_S)}{0}} = \{\epsilon\} \subseteq \tau(H)$.
Let $e \in \bo{(\widehat{K}_S)}{n+1}$ and $d = \bo{e}{n}$.
Since $d \in \bo{(\widehat{K}_S)}{n}$, $d \in \tau(H)$ by induction hypothesis.
Therefore, $|d| \leq g(n)$ and $d' = d\bot^k \in H$ for $k = g(n) - |d|$.
We have $\exbarS^{|e|}(e) \ne \emptyset$ because $e \in \widehat{K}_S$,
and $\exbarS^{|e|}(e) \subset \exbarS^{|e|-1}(e) = \exbarS^{|e|-1}(d')$.
Therefore, $k_{d'} \geq |e|$ and thus $g(n+1) \geq |e|$.
Therefore, $e\bot^k \in H$ for $k = g(n+1) - |e|$.

$H$ is decidable because,
for a $\{0,1,\bot\}$-sequence $e$, $e \in H$ if and only if $|e| = g(\mathrm{level}(e))$.

\end{proof}
\end{lemma}

In the following, we fix such a function
$g : \mathbb{N} \to \mathbb{N}$ and define $H$ as in Lemma \ref{lemma:H}.
We identify $e \in \widehat{K}_S$ with $e\bot^{g(\mathrm{level}(e)) - |e|}
\in H$, and
we consider that $\widehat{K}_S \subseteq H$.
Each element of  $\bo{H}{n}$ contains $n$ digits and $g(n) - n$ copies of $\bot$.  Therefore, $|\bo{H}{n}| = {}_{g(n)}C_{n}2^n$.
In addition,  every  $d \in  \bo{H}{n}$ has the same number of successors
$|\mathrm{succ}_H(d)| = 2(g(n+1) - n)$.
Since $H$ is finitely branching and tree-like,
by applying Proposition \ref{prop:tree-like}, we have a strictly increasing function $f$
and a partial surjective map $\beta:\subseteq \2^* \to H$ such that
$\2^{f(n)}$ is mapped to $\bo{H}{n}$.
In this case, $f$ is the function $f(n) = \Sigma_{k=0}^{n-1} {\lceil\log{|g(k+1)-k}|\rceil}$.
$\beta$ maps a tree in $\dom(\beta) \subseteq \2^*$ to a tree-like subset of  $H$,
and injectively maps an infinite path in $\dom(\beta) $ to an infinite path in $H$.
By using some bijection $\mathbb{N} \to H$, we can represent a tree-like subset of $H$ by its characteristic function, and we shall denote the represented space of tree-like subsets of $H$ by $\treeL$.
Thus, $\beta$ induces a computable partial map from $\tree$ to $\treeL$, which we also denote by $\beta$.


\begin{example}

  For the case of the Gray subbase $G$ in Example \ref{ex:gray}, $g(n) = n+1$ satisfies Lemma \ref{lemma:H}
and $f(n) = 2n$ satisfies Proposition \ref{prop:tree-like}.
Therefore,
$|\bo{H}{n}| = (n+1)2^n$ and $|\mathrm{succ}_H(e)| = 4$ for every  $e \in H$.
$\bo{H}{1}$ contains $\bot 0$ in addition to the three elements of
$\bo{(\widehat{K}_G)}{1}$.
$\beta$ injectively maps $\2^{2}$ to $\bo{H}{1}$, and
maps $2^{4}$ elements of $\2^{4}$ to $2^{3}$ elements of $\bo{H}{2}$.
\end{example}

\begin{lemma}\label{lemma:lub}
\begin{enumerate}
\item   $(H, \preceq)$ is a partially ordered set.

\item   $(H, \preceq)$ is consistently complete.  That is, every pair $(d, e)$ of consistent elements have
  their least upper bound $d \sqcup e$ in $H$.  We have $S(d \sqcup e) = S(d) \cap S(e)$.

\item  $S(d) = \cup \{S(e) \mid e \in \mathrm{succ}_H(d) \}$ for every $d \in H$.
  More generally,
$S(d) = \cup \{S(e) \mid e \in \bo{H}{n}, d \prec e \}$ for every $d \in H$ and $n >  \mathrm{level}(d)$.

\item  Every subset $A$ of $H$ has a greatest lower bound $\sqcap A$.
\end{enumerate}
\begin{proof}
\begin{enumerate}
\item  For $d, e \in H$,  if  $d\preceq e$ and $e\preceq d$,
then $d$ and $e$ have the same level $n$ and therefore
have the same length $g(n)$.

\item In $\tau(\T^*)$, if $d \comp e$, then their least upper bound $d \sqcup e$ exists.   Therefore, for $d, e \in H$,
   if $|\tau(d)| \leq |\tau(e)|$, then
   $|\tau(d) \sqcup \tau(e)| = |\tau(e)| \leq |e| = g(\mathrm{level}(e)) \leq g(\mathrm{level}(d \sqcup e))$.
   Thus, least upper bound $d \sqcup e$ of $d, e \in H$ is obtained by adding some
   $\bot$ at the end of $\tau(d) \sqcup \tau(e)$.

\item For $d \in \bo{H}{n}$,
$S(d) \supseteq \cup \{S(e) \mid e \in \mathrm{succ}_H(d) \}$ because $S(d) \supseteq S(e)$ for each $e \in \mathrm{succ}_H(d)$.
$S(d) \subseteq \cup \{S(e) \mid e \in \mathrm{succ}_H(d) \}$
holds if $\tau(d) \in \widehat{K}_{S}$ by Proposition \ref{prop:compact} (3).
It holds for the case $\tau(d) \not \in \widehat{K}_{S}$ because
$\exbarS(e) $ is empty for $e$ such that $d \prec^1 e$ and $|d| < |e|$.

\item From (2), every bounded subset of $H$ has a least upper bound.
   Therefore, the set of lower bounds of $A$, which is bounded by an element of $A$, has a least upper bound, which is a greatest lower bound of $A$.
 \end{enumerate}
\end{proof}
\end{lemma}

Note that Lemma \ref{lemma:lub} (2) does not imply that
$S(d) \cap S(e) \ne \emptyset$ if $d \comp e$,
because
$S(e)$ can be empty for $e \in H$ though $S(e)$ is not empty for $e \in \widehat{K}_{S}$.

\subsection{Main theorem}

Now we obtain results corresponding to Corollary \ref{corr:nicefunctionscantor}  for CCCMS.

Let $X$ be a CCCMS.
By Theorem \ref{theo:dyadicexists}, $X$ admits a proper computable dyadic subbase $S$.  Let $H$, $f$, $\beta$
and $\treeL$ be as defined in the previous subsection for $X$ and $S$.
We define a subset $\mathcal{N} \subseteq \treeL$ as
\[\mathcal{N} = \{E \in \treeL \mid L(E) \subseteq L(\widehat{K}_S)\}\,.\]

\begin{lemma}\label{lem:rhobracket}
  If $E \in {\mathcal N}$, then  $[E]_S = \{\rho(p) \mid p \in L(E)\}$.
\end{lemma}
\begin{proof}
  By Proposition \ref{prop:bracket},
  \ref{prop:compact}(4). and \ref{prop:rho}(3).
\end{proof}

\begin{proposition}\label{prop:nicefunctionsgeneral}
Let $\X$ be a CCCMS.
There is a surjection $\tilde{t}_{\X} : \treeL \to {\A}(\X)$ and a multivalued map $\tilde{s}_{\X}: {\A}(\X) \rightrightarrows \treeL$ such that
\begin{enumerate}
\item $\tilde{s}_{\X}$ and $\tilde{t}_{\X}$ are computable.
\item $\tilde{t}_{\X} \circ \tilde{s}_{\X} = \id_{\A(\X)}$.
\item For each $A \in {\A}(\X)$, $\tilde{s}_{\X}(A) \in \mathcal{N}$.
\end{enumerate}

\begin{proof}
For a finite subset $C \subset H$ and $A \subseteq \X$, we simply say that $C$ is a covering of $A$ if
$S(C) = \{S(e) \mid e \in C\}$ is a covering of $A$.

For $E \in \treeL$, we define $\tilde{t}_\X(E) = [E]_S =
\bigcap_{n \in \mathbb{N}} \bigcup_{e \in \bo{E}{n}} \overline{S}(e)$.
As $\A(\X)$ is effectively closed under finite unions and countable intersections,
and since $e \notin E$ is recognizable, this does define a computable function.

For $A \in \A(\X)$, we define $\tilde{s}_\X(A)$ as follows.
We inductively define a computable strictly monotonic function $h : \N \to \N$ and
a sequence $E_0,E_1,\ldots$ of subsets $E_n \subseteq \bo{H}{h(n)}$ such that
$\forall d \in E_{n+1} \exists  e \in E_{n} \ d \succ e $ and that
$E_n$ is a covering of $A$.
Then, we define a tree-like subset $E \subseteq H$ such that
$\bo{E}{h(n)} = E_n$ by defining for $h(n) < m < h(n+1)$,
$\bo{E}{m} =
\{c \in \bo{H}{m} \mid \exists d \in E_{n+1} \exists  e \in E_{n} \ e \prec c \prec d\}$. Finally, we set $\tilde{s}_\X(A) = E$.


We start with $h(0) = 0$ and $E_0 = \{\varepsilon\}$.
We search for (in a dove-tailing way) finite coverings $C \subset H$ of $A$.
There are infinitely many such coverings, and
as $\X$ is computably compact, we will find each such covering eventually.

If we do find a new covering $C$ at stage $n+1$, then we consider
$$E_n \oplus C := \{ d \sqcup e \mid d \in E_n \wedge e \in C \wedge d \comp e\}.$$  Here,  $d \sqcup e$ exists because of Lemma \ref{lemma:lub} (2).
Because $S(d \sqcup e) = S(d) \cap S(e)$,
$(\cup S(E_n)) \cap (\cup S(C)) = \cup S(E_n \oplus C)$.  That is,  $E_n \oplus C$ is
a covering of $A$, which is a refinement of both coverings $E_n$ and $C$.

We then choose $h(n+1)$ so that $E_n \oplus C \subseteq \bo{H}{\leq h(n+1)}$ 
and set
\[
E_{n+1} := \{e \in \bo{H}{h(n+1)} \mid (\exists d \in E_n \oplus C \ d \preceq e) \land S(e) \ne \emptyset\}.
  \]
Note that $S(e) \ne \emptyset$ is decidable  by Proposition \ref{prop:decidable}.
Because of Lemma \ref{lemma:lub} (3), for $d \in E_n \oplus C$,
we have $S(d) = \cup\{S(e) \mid e \in \bo{H}{h(n+1)} \land d \preceq e\}$.
Therefore,
we have $\cup S(E_n \oplus C) = \cup S(E_{n+1})$.
Thus, $E_{n+1}$ is also
a covering of $A$, which is a refinement of both coverings $E_n$ and $C$.

This procedure ensures that
$\tilde{t}_\X(\tilde{s}_\X(A)) =
\tilde{t}_\X(E) = \bigcap_{n \in \mathbb{N}} \bigcup_{e \in \bo{E}{n}} \overline{S}(e)$
is equal to the intersection of all finite basic open coverings of $A$, which is equal to $A$ in a compact Hausdorff space.  Therefore, we have $\tilde{t}_\X(\tilde{s}_\X(A)) = A$.


Finally, we show that $L(\tilde{s}_{\X}(A)) \subseteq L(\widehat{K}_S)$.
Let $e_0 \prec e_1 \prec \ldots$ be an infinite sequence such that
$e_n \in E_n$ and let $p \in L(\widehat{K}_S)$ be its limit.
We first show that $|\bar{S}(p)| = 1$.
We have $\bar{S}(p) = \cap_{n \in \N} \bar{S}(e_n)$ by the definition of $\bar{S}$.  Therefore, if $\bar{S}(p) = \emptyset$, then $\bar{S}(e_n) = \emptyset$ for some $e_n$.   It contradicts the construction because $S(e_n) \ne \emptyset$.   Suppose that $|\bar{S}(p)| > 1$ and $x \ne y \in \bar{S}(p)$.
Since $X$ is computably regular by Theorem \ref{theorem:coregular} and computably compact,
there is a finite covering $C$ that does not contain an element $e$ such that $\{x, y\} \subset  \bar{S}(e)$.
If such a covering is found at stage $n$,  then,  no element of $\bar{S}(E_n)$ contains both $x$ and $y$.
It contradicts to $\bar{S}(e_n) \supseteq \bar{S}(p) \supseteq \{x, y\}$.  Thus, we have
$|\bar{S}(p)| = 1$ and suppose that $\bar{S}(p) = \{x\}$.  We show $\varphi_S(x) \preceq p$.
Suppose that  $x \in S_{i, a}$ for $i \in \N$ and $a \in \{0, 1\}$. Then,
since $X$ is computably regular and computably compact, there is a finite covering $C$ of $X \setminus S_{i, a}$ such that $\bar{S}(e) \not \ni x$ for every $e \in C$.
We can assume that $e(i) \ne a$ for every $e \in C$, because
$e(i) = a$ means $S(e) \subseteq S_{i, a}$ and  we can remove such elements from $C$ to have a covering of
$X \setminus S_{i, a}$.
Then, $C \cup \{\bot^i a\}$ is a covering of $X$ and therefore is also a covering of $A$.
Suppose that such a covering is found at stage $n$.  Then, for every element $e$ of $E_n$,
$x \in \bar{S}(e)$ implies $e = \bot^ia$.  Since $x \in \bar{S}(e_n)$, we have $e_n = \bot^i a$ and
we have $p(i) = a$ because $p \succ e_n$.
Since it holds for every $(i, a)$ such that $x \in S_{i, a}$, we have $\varphi_S(x) \preceq p$.
\end{proof}
\end{proposition}

\begin{theorem}
\label{theo:hereditary}
Let $\X$ be a CCCMS.
There is a hereditary representation of $\mathcal{K}(\X)$.
\begin{proof}
  First, we define a computable function $\gamma$ from $\treeL$ to $\tree$, which is a right inverse of $\beta:\tree \to \treeL$.
  For  $E \in \treeL$, we define $T_n  = \beta^{-1}(\bo{E}{n}) \subseteq \2^{f(n)}$.
We have $\forall w \in T_{n+1} \exists v \in T_n v \sqsubseteq w$ and
therefore it induces a tree $T$ such that $T \cap \{0,1\}^{f(n)} = T_n$.  We define $\gamma(E) = T$.
We have $\beta \circ \gamma = \mathrm{id}_{\treeL}$.

With the maps $\tilde{t}_\X :\subseteq \treeL \to \A(\X)$ and $\tilde{s}_{\X}: {\A}(\X) \rightrightarrows \treeL$
we defined in Proposition \ref{prop:nicefunctionsgeneral},
we define $t_\X :\subseteq \T^\omega \to \A(\X)$ as
$\tilde{t}_\X \circ \beta \circ \textrm{Prune}^{-1} \circ \delta_\pruned$.
Note that
$\tilde{t}_\X(\beta(T)) = \tilde{t}_\X(\beta(T'))$ if $\textrm{Prune}(T) = \textrm{Prune}(T')$ as we observed
after Proposition \ref{prop:tree-like}.  Therefore,
the value of $t_\X$ is determined not depending on which value of $\textrm{Prune}^{-1}$
is used.  Let $s_\X : \A(\X) \mto \T^\omega$ be its multi-valued inverse
computed by the realizer of $\textrm{Prune} \circ \gamma \circ \tilde{s}_\X$.
We have $t_\X \circ s_\X = \textrm{id}_{\A(X)}$.
We define the representation $\psi_\X : \T^\omega \to \A(\X)$ by restricting
the domain of $t_\X$ to the $\delta'$-names of $(\textrm{Prune}\circ \gamma)(\mathcal{N})$.
Note that since $\tilde{s}_\X$ is a map to $\mathcal{N}$
by Proposition \ref{prop:nicefunctionsgeneral}(3),
$s_\X$ is a map to the domain of $\psi_\X$, and therefore $\psi_\X \circ s_\X = \textrm{id}_{\A(X)}$.




We show that $\psi_\X$ is hereditary. That is, if
$\psi_\X(p) = A$ and $B$ is a non-empty compact subset of $A$,
then $\exists q \sqsupseteq p\  (\psi_\X(q) = B)$.
Let $E \in \mathcal{N}$ satisfy $A = \tilde{t}_\X(E)$,
$T = \gamma(E)$, and $p$ is a $\delta'$-name of  $\textrm{Prune} (T)$.
As the hereditarity of $t_\Cantor$ shows,
if $T' \subseteq T$, then there exists $q \sqsupseteq p$ such that
$\delta'_\pruned(q) = \textrm{Prune}(T')$.
Since $\gamma : \treeL \to \tree$ is monotonic in addition, we only need to show that
there exists a tree-like subset $F \subseteq E$ such that $F \in \mathcal{N}$ and $B = \tilde{t}_\X(F)$.
We define $F = \{e \in E \mid \bar{S}(e) \cap B \ne \emptyset\}$.
$F$ is a tree-like subset of $E$ and therefore $L(F) \subseteq L(\widehat{K}_S)$.
We need to show that $[F]_S= B$.
By Lemma \ref{lem:rhobracket}, we show
$\{\rho(p) \mid p \in L(F)\} = B$.
If $x \in B$, then $p = \varphi(x)$ satisfies $x = \rho(p)$ and we have $\bo{p}{\leq n} \in F$ for every $n \in \N$.
Suppose that $p$ is the limit of a path $(e_n)_{n \in N}$ in $F$. Since $\bar{S}(e_n) \cap B \ne \emptyset$ for each $n$,
we have $\bar{S}(p) \cap B \ne \emptyset$ and thus $\rho(p) \in B$ by Proposition \ref{prop:rho}(3) and
Lemma \ref{lem:rhobracket}.
\end{proof}
\end{theorem}

We finally modify $\psi_\X$ to form a faithful representation $\psi'_\X$ of $\K(\X)$.
For this purpose,
we modify $\gamma$ in the proof of Theorem \ref{theo:hereditary}
and construct $\gamma' : \treeL \to \tree$
so that, for a tree-like set $E$ such that $[E]_S$ is a singleton,
$\gamma'(E)$ is an infinite path $T$ such that $\beta(T) \subseteq E$.

\begin{example}
  As the right half of Figure \ref{fig:G} shows, $\widehat{K}_G$ contains a tree-like set $E$ such that
$L(E) = \{1^\omega\}$ (and therefore $[E]_S = \{1/3\}$)  and $|\bo{{E}}{n}| = 2$ for every $n \in \N^+$.
 As the left half of Figure \ref{fig:G} shows, $\widehat{K}_G$ contains a tree-like set $F$ such that
$L(F) = \{\bot 1 0^\omega, 0 1 0^\omega, 1 1 0^\omega \}$ (and therefore $[F]_S = \{0\}$)  and $|\bo{{F}}{n}| = 3$ for every $n \in \N^+$.
$\gamma(E)$ and $\gamma(F)$ are binary trees which contain infinite number of paths
though $[E]_S$ and $[F]_S$ are singletons.
The new function $\gamma'$ we will define chooses one infinite path from these tree-like sets, and converts them to infinite paths in $\{0,1\}^*$.
\end{example}

\begin{lemma}\label{lem:sqcapsingle}
  Let $S$ be a proper computable dyadic subbase of a CCCMS $\X$ and
  let $E \in \mathcal{N}$ be a tree-like set such that $[E]_S$ is a singleton.
  For each $n \in \N$  there is $m \in \N$ such that $\textrm{level}(\sqcap \bo{E}{m}) \geq n$.

  \begin{proof}
    Let $\{x\} = [E]_S$, $e = \bo{\varphi_S(x)}{n}$, $i \in \dom(e)$ and $F_i = \{d \in E \mid \textrm{level}(d) < i\linebreak \lor\ d(i) \ne e(i)\}$.
    For each $i$, $F_i$ is a tree-like subset of $E$.  If $F_i$ is infinite, then  there is an infinite path in $F_i$ because $H$ is finite-branching.
    Therefore, there is $q \in L(F_i)$ such that $q(i) \ne e(i)$ and it contradicts the fact that all the elements of
$L(E)$ are greater than or equal to  $\varphi_S(x)$.
    Therefore, $F_i$ is a finite set and there exists $m_i$ such that $d(i) = e(i)$ for every element of $d \in \bo{E}{m_i}$.   Let $m$ be the maximal of $\{m_i \mid i \in \dom(e)\}$.
    Then, for every  $d \in \bo{E}{m}$, $d \succeq e$.  Therefore, $\textrm{level}(\sqcap \bo{E}{m}) \geq n$.
  \end{proof}
\end{lemma}

This lemma shows that the sequence $(\sqcap \bo{E}{m})_{m \in \N}$ can be extended to an infinite path whose limit $p$ satisfies $\bar{S}(p) = [E]_S$.

\begin{theorem}\label{theo:minimalfaithful}
  Let $\X$ be a CCCMS. There is a faithful $\T^\omega$-representation of $\mathcal{K}(\X)$.
\begin{proof}
  In this proof, we  define another computable right inverse $\gamma': \treeL \to \tree$ of $\beta: \tree \to \treeL$ such that if $|[E]_S| = n$ then $\gamma'(E)$ is a tree with $n$ infinite paths.



(i) Definition of $\gamma'$.

Suppose that $E \in \treeL$ is given.  We construct $T = \gamma'(E)$ by
inductively defining $T_n \subseteq \{0,1\}^{f(n)}$ for $n \in \N$.
$T_n$ satisfies the following condition.
 Let $D_0, \ldots, D_{j-1}$ be the
division of $\bo{E}{n}$ into equivalence classes with respect to the transitive closure of $\comp$,
and $d_0, \ldots, d_{j-1}$ be their greatest lower bounds.  The condition is that
$T_n$ is divided into  $T_{n,i}\,  (0 \leq i < j)$  such that (1) $\beta(T_{n,i}) = D_i$,
(2)  the greatest lower bound (i.e., longest common prefix) $w_i$ of $T_{n,i}$ satisfies
$\beta(w_i) = d_i$,   and (3) $T_{n,i}  =\, \uparrow\! w_i \cap \beta^{-1}(D_i)$.


First, set $T_0 = \{\epsilon\}$.  Suppose that $T_n \subseteq \{0,1\}^{f(n)}$ is defined and it satisfies the above condition.  Let $D_0, \ldots, D_{j-1}$ be the division of
$\bo{E}{n+1}$ into equivalence classes with respect to the transitive closure of $\comp$, and
  $d_0, \ldots, d_{j-1}$ be their greatest lower bounds.  Since $d \comp e$, $d' \prec d$ and
  $e' \prec e$ impliy $d' \comp e'$, for each $D_i$, there is an equivalence class $D'$ of $\bo{E}{n}$
  such that every element of $D_i$ is greater than an element of $D'$.
Let $d'$ be the greatest lower
  bound of $D'$. By induction hypothesis, for the greatest lower bound $w'$ of $T_n \cap \beta^{-1}(D')$, we have $\beta(w') = d'$.  On the other hand, we have $d' \preceq d_i$ because $d'$ is a lower bound
  of $D_i$.  Therefore, we can choose $w_i\sqsupseteq w'$
such that $\beta(w_i) = d_i$ by Proposition \ref{prop:tree-like}.
Now, we select such a $w_i$ for each $i < j$, and
define $T_{n+1, i} = \uparrow\! w_i \cap \beta^{-1}(D_i)$.  Then, we define
$T_{n+1} = \cup_{i < j}  T_{n+1, i}$.

We show that  $\beta(T_{n+1,i}) = D_i$.  For $e \in D_i$, since $d_i \preceq e$ and $\beta(w_i) = d_i$,
there exists $v \in T$ such that $w_i \sqsubseteq v$ and $\beta(v) = e$ again by Proposition \ref{prop:tree-like}.
Since $\beta(v)  \in D_i$, we have $v \in T_{n+1, i}$ by condition (3) of $\gamma'$.  Next, we show that
$\forall w \in T_{n+1} \exists v \in T_n\  v \sqsubseteq w$.    Set $v = w_{\leq f(n)}$.
Suppose that $w \in T_{n+1, i}$.
Since $w'$ is a prefix of $w_i$ and $w_i$ is a prefix of $w$, $w'$ is a prefix of $w$.
Therefore, $w'$ is a prefix of $v$ and thus we have $v \in T_n$.
Thus, $(T_n)_{n \in \N}$ induces a
  tree $T$ such that $T \cap \{0,1\}^{f(n)} = T_n$, and we define $\gamma'(E) = T$.
One can see from this construction that
  $\gamma': \treeL \to \tree$ is a computable map.
We have  $\beta \circ \gamma' = \textrm{id}_{\treeL}$ because $\beta(T_n) = \bo{E}{n}$.

(ii) If $|[E]_S| = n$, then $\gamma'(E)$ is a tree with $n$ infinite paths.

We first study the case that $[E]_S$ is a singleton.
By Lemma \ref{lem:sqcapsingle}, for each $n$, there is $m$ such that the level of $\sqcap \bo{E}{m}$ is greater or equal to $n$.
Among the equivalence classes of $\bo{E}{n}$
with respect to the transitive closure of $\comp$, only one of them, which we call $D_n \subseteq \bo{E}{n}$, contains an element less than $\sqcap \bo{E}{m}$.  It means that
only elements in $D_n$ can be extended to infinite paths.

From the definition of $\gamma'$,  there is $w_n$ such that $\beta(w_n) = \sqcap D_n$.
Note that $w_n$ is the only element in $\beta^{-1}(\sqcap D_n)$ that can be extended to an infinite path in $\gamma'(E)$
beause $T_n \cap \beta^{-1}(D_n) \subseteq \uparrow w_n$ by condition (3) of $\gamma'$.
Therefore, $w = (w_n)_{n \in \N}$ is extended to the only infinite path in $T$.

If $[E]_S$ contains $k$ points, then there exists $N$ such that for each $n \geq N$,
$\bo{E}{n}$ is divided into $k$ equivalence classes  with respect to the transitive closure of $\comp$.  Therefore, $T$ has $k$ infinite paths with a similar argument.

(iii) Definition of $\psi'_\X$.

We define the multi-valued map $s_\X' : \A(\X) \mto \T^\omega$ through  the realizer of $\textrm{Prune} \circ \gamma' \circ \tilde{s}_\X$.  Set
$${\cal O} = \{p \in \T^\omega \mid \mbox{$|t_\X (p)|$ is infinite or
$|t_\X (p)| = n$ and $p$ contains $n$ copies of $\bot$} \}.$$

Suppose that $|A| = n$.  Then, for $E = \tilde{s}_\X(A)$,  $\gamma'(E)$ is a tree with $n$ paths by (ii).  Therefore, $p = s_\X'(A)$ has $n$ copies of $\bot$ and $|t_\X (p)| = n$.   Therefore, $p \in {\cal O}$.
Suppose that $|A|$ is infinite. Since $t_\X ( s_\X'(A) = A$, we have $s_\X'(A) \in {\cal O}$.
Thus, $s_\X' $ is a computable map to ${\cal O}$.
Now, we define the representation $\psi'_\X :\subseteq \T^\omega \to \A(\X)$ as the restriction of $\psi_\X$ to ${\cal O} \cap (\textrm{Prune} \circ \gamma) (\mathcal{N})$.

(iv) $\psi'_\X$ is faithful.

First, we show that
condition (1) and (2) of  Definition \ref{def:faithful} are satisfied
if $A = \psi'_\X(p)$ is a $n$-point set.
(1) is satisfied since $\psi'_\X$ is restricted to  ${\cal O}$.
(2) is satisfied because all the infinite paths of a tree-like set $E \in {\cal N}$ are in ${\cal N}$, and
that $\psi_\Cantor$ is a faithful map.
In order to show that $\psi_\X'$ is hereditary, compared with the proof that
$\psi_\X$ is hereditary in Theorem \ref{theo:hereditary},
we only need to see that if $A = \psi'_\X(p)$ is an infinite set and $B \subset A$ is a $n$-point set, then there is $q \prec p$ such that $B = \psi'_\X(q)$.
It is done by selecting, for each element of $B$,  one infinite path from $\delta_\pruned(p)$.
\end{proof}
\end{theorem}

\section{Applications to finite closed choice}
\label{sec:weihrauch}
As mentioned in the introduction, an initial goal of our explorations was to better understand and to generalize a construction from \cite{paulyleroux}. The construction was used to show that the difficulty of selecting a point from a closed subset of cardinality equal to (at most) $n$ from a computably rich computably compact computable metric space $\mathbf{X}$ does not depend on the choice of $\mathbf{X}$. The original construction crucially depended on knowledge of $n$, and thus did not extend to the task of selecting a point from a finite closed set. Based on Theorem \ref{theo:minimalfaithful}, we can fill this gap.

These considerations belong to the investigation of closed choice principles in the Weihrauch lattice. Closed choice principles are used to calibrate non-computable tasks where (provided with a suitable certificate) incorrect solutions can be effectively rejected, and some guarantees on the set of solutions are given. In our context, these guarantees are that the solutions belong to compact metric spaces, and that each instance has only finitely many correct solutions (plus certificates). For an survey of Weihrauch reducibility, and the definition of all terms left undefined here, we refer to \cite{pauly-handbook}. The important role of close choice principles was noted in \cite{paulybrattka}.

\begin{definition}
For a represented space $\mathbf{X}$, closed choice $\C_\X : \subseteq \mathcal{A}(\X) \mto \X$ is defined via $x \in \C_\X(A)$ iff $x \in A$. By $\C_{\X,\sharp \leq n}$ and $\C_{\X,\sharp=n}$ we denote the restriction of $\C_\X$ to sets of cardinality up to $n$ and equal to $n$ respectively. By $\C_{\X,\sharp<\infty}$ we denote the restriction of $\C_\X$ to finite sets.
\end{definition}

\begin{definition}
\label{def:concretize}
Let $\operatorname{Concretize} : \T^\omega \mto \Cantor$ be defined via $q \in \operatorname{Concretize}(p)$ iff $q \succeq p$. Let $\operatorname{Concretize}_{<\infty}$ be the restriction of $\operatorname{Concretize}$ to $\{p \in \T^\omega \mid |\{n \mid p(n) = \bot\}| < \infty\}$.
\end{definition}

\begin{corollary}
Let $\mathbf{X}$ be a computably rich CCCMS. Then
\[\C_{\mathbf{X},\sharp<\infty} \equivW \C_{\Cantor,\sharp<\infty} \equivW \operatorname{Concretize}_{<\infty}\]
\begin{proof}
\begin{description}
\item[$\C_{\mathbf{X},\sharp<\infty} \leqW \operatorname{Concretize}_{<\infty}$] By using a faithful representation for $\K(\X) \cong \A(\X)$ from Theorem \ref{theo:minimalfaithful}, 
 we can apply $\operatorname{Concretize}_{<\infty}$ to the input for $\C_{\mathbf{X},\sharp<\infty}$ and obtain a singleton subset. By admissibility of $\X$, we can then extract the point from the singleton.
\item[$\operatorname{Concretize}_{<\infty} \leqW \C_{\Cantor,\sharp<\infty}$] Given $p \in \mathbb{T}^\omega$ we can compute $\{q \in \Cantor \mid q \succeq p\} \in \mathcal{A}(\Cantor)$. The claim follows from the definition of $\operatorname{Concretize}_{<\infty}$ in Definition \ref{def:concretize}.
\item[$\C_{\Cantor,\sharp<\infty} \leqW \C_{\mathbf{X},\sharp<\infty}$] As $\X$ is computably rich, by definition it contains a copy of $\Cantor$ as a computable closed subspace.
\end{description}
\end{proof}
\end{corollary}

By $\mathrm{AoUC}_{\X}$ (all-or-unique choice) we denote the restriction of $\C_\X$ to $\{A \in \mathcal{A}(\X) \mid |A| = 1 \vee A = X\}$. The degree $\mathrm{AoUC}_{[0,1]}$ was studied first in \cite{paulyincomputabilitynashequilibria}, where it was found to relate to the complexity of finding Nash equilibria in bimatrix games. Further results on $\mathrm{AoUC}_{\X}$ are available in \cite[Section 16 \& 17]{hoelzl} and \cite{pauly-kihara2-mfcs}.

\begin{proposition}
$\mathrm{AoUC}_{[0,1]} \nleqW \C_{\Cantor,\sharp<\infty}$.
\begin{proof}
Note that $\C_{\Cantor,\sharp<\infty}$ is a cylinder, so if $\mathrm{AoUC}_{[0,1]} \leqW \C_{\Cantor,\sharp<\infty}$, then already $\mathrm{AoUC}_{[0,1]} \leq_{\mathrm{sW}} \C_{\Cantor,\sharp<\infty}$. Assume this would hold. Consider the input $\uint$ to $\mathrm{AoUC}_{[0,1]}$. This will be mapped to some finite set $F \in \mathcal{A}(\Cantor)$. The outer reduction witness $K : \subseteq \Cantor \to \uint$ can be lifted to compact sets. Let $H$ be the inner reduction witness. From $A \in \dom(\mathrm{AoUC}_{[0,1]})$ we can compute $K[H(A)] \in \mathcal{A}(\uint)$, in particular $K[H(\uint)]$. Since $|K[H(\uint)]| \leq |F|$, at some finite time some non-empty open ball $B \in \mathcal{O}(\uint)$ is removed from $K[H(\uint)]$. But at that moment, we can still alter the name for $\uint$ to a name for $\{x\}$ for some $x \in B$, and the reduction cannot adapt. Thus, $\mathrm{AoUC}_{[0,1]} \leqW \C_{\Cantor,\sharp<\infty}$ cannot hold.
\end{proof}
\end{proposition}

By $\cc_{\X}$ we denote the restriction of $\C_\X$ to connected sets. Clearly, for a connected space $\X$ we find that $\mathrm{AoUC}_\X \leqW \cc_\X$. The degrees $\cc_{[0,1]^n}$ are equivalent to Brouwer's fixed point theorem for $[0,1]^n$ as shown in \cite{paulybrattka3}.

\begin{corollary}
$\cc_{\uint} \nleqW \C_{\Cantor,\sharp<\infty}$.
\end{corollary}

\subsection{Generalized register machines}
Building on Tavana and Weihrauch's \cite{tavana}, Neumann and Pauly \cite{paulyneumann} introduced the operator $^\diamond$ on Weihrauch degrees. Roughly spoken, $f^\diamond$ is the universal problem for register machines using computable operations and tests, together with $f$ as primitive operations. Thus, we can view a Weihrauch reduction to $f^\diamond$ as a uniformly computably procedure that solves the problem by making finitely many oracle calls to $f$. This notion was studied as generalized Weihrauch reducibility by Hirschfeldt and Jockusch \cite{hirschfeldt,hirschfeldt2}.

\begin{proposition}
\label{prop:cinftydiamond}
$\C_{\Cantor,\sharp<\infty} \equivW \C_{\Cantor,\sharp<\infty}^\diamond$.
\begin{proof}
Given some input $(M,x)$ to $\C_{\Cantor,\sharp<\infty}^\diamond$, we make an oracle guess $p \in \Cantor$. Then we simulate the generalized register machine $M$. Each time $M$ calls its oracle $\C_{\Cantor,\sharp<\infty}$, we split the oracle into two parts $\langle p',q_i\rangle$ and use $q_i$ as putative output of  $\C_{\Cantor,\sharp<\infty}$, and $p'$ as new oracle. Once $M$ terminates, we check that the current oracle is equal to $0^\omega$. In addition, we check that all outputs of $\C_{\Cantor,\sharp<\infty}$ were guessed correctly.

We need to argue that there are finitely many valid guesses. As each oracle use has only finitely many correct answers, and since every computation of $M$ is finite, the potential computation paths of $M$ form a finitely branching tree without infinite paths, hence a finite tree. Once all oracle calls are taken care of, only $0^\omega$ remains as the unique correct guess.
\end{proof}
\end{proposition}

\begin{proposition}
\label{prop:cleqdiamond}
\[\C_{\Cantor,\sharp \leq 2}^\diamond \equivW \coprod_{n \in \mathbb{N}} \C_{\Cantor,\sharp \leq n} \equivW \C_{\Cantor,\leq 2}^*\]
and
\[\C_{\Cantor,\sharp = 2}^\diamond \equivW \coprod_{n \in \mathbb{N}} \C_{\Cantor,\sharp = n} \equivW \C_{\Cantor,= 2}^*\]
\begin{proof}
In each case, the second equivalence was already shown in \cite{paulyleroux}, and the reduction from the right-most degree to the left-most is trivial. We thus only need to show that $\C_{\Cantor,\sharp \leq 2}^\diamond \leqW \coprod_{n \in \mathbb{N}} \C_{\Cantor,\sharp \leq n}$ and $\C_{\Cantor,\sharp = 2}^\diamond \leqW \coprod_{n \in \mathbb{N}} \C_{\Cantor,\sharp = n}$.

Similar to the argument in Proposition \ref{prop:cinftydiamond}, consider the computation tree of a generalized register machine $M$ making oracle calls to $\C_{\Cantor,\sharp \leq 2}$ (respectively $C_{\Cantor,\sharp = 2}$). We know that the branching factor of the computation tree at the oracle calls is exactly two (respectively at most two). By Weak K\"onig's Lemma, the depth of the computation tree is bounded, and we can in fact effectively compute some upper bound $t$ on its depth. We then use $n := 2^t$ as first component of the input to $\coprod_{n \in \mathbb{N}} \C_{\Cantor,\sharp \leq n}$ (respectively $2^{2^t}$ as input to $\coprod_{n \in \mathbb{N}} \C_{\Cantor,\sharp = n}$).

In the former case, we can simply feed the same set of guesses constructed in the proof of Proposition \ref{prop:cinftydiamond} to $\C_{\Cantor,\sharp \leq 2^t}$, as the upper bound is guaranteed to be correct, and that is all that is required. In the latter case, we keep refining the computation tree. Each time we learn that one of the up to $2^t$ many oracle calls in the computation tree is not actually happening, we duplicate all remaining points in the input set. This ensures that in the end, we produce a set of the correct cardinality.
\end{proof}
\end{proposition}

\begin{corollary}
$\C_{\Cantor,\sharp = 2}^\diamond \leW \C_{\Cantor,\sharp \leq 2}^\diamond \leW \C_{\Cantor,\sharp<\infty}$
\begin{proof}
The first reduction is trivial. That it is strict follows from $\C_{\Cantor, \sharp \leq 2} \nleqW \C_{\Cantor, \sharp =n}$ shown in \cite{paulyleroux} and Proposition \ref{prop:cleqdiamond}. The second reduction follows from Proposition \ref{prop:cinftydiamond}. We can see that it is strict by first observing that $\C_{\Cantor,\sharp<\infty} \leqW \C_{\Cantor,\sharp \leq 2}^\diamond$ would imply $\C_{\Cantor,\sharp<\infty} \leqW \coprod_{n \in \mathbb{N}} \C_{\Cantor,\sharp \leq n}$. As $\C_{\Cantor,\sharp<\infty}$ is a fractal, that in turn implies $\C_{\Cantor,\sharp<\infty} \leqW \C_{\Cantor,\sharp \leq n}$ for some $n \in \mathbb{N}$. But that is a contradiction to $\C_{\Cantor,\sharp \leq n+1} \nleqW \C_{\Cantor,\sharp \leq n}$ established in \cite{paulyleroux}.
\end{proof}
\end{corollary}
\subsection{A digression on and comparison with $\mathrm{Sort}$}
\begin{definition}
Let $\mathrm{Sort}_* : \subseteq \Baire \to \Baire$ be defined via $p \in \dom(\mathrm{Sort}_*)$ iff $\exists k \ |\{n \mid p(n) = k\}| = \infty$, and $\mathrm{Sort}_*(p) = 0^{c_0}1^{c_1}\ldots k^\infty$, where $|\{n \mid p(n) = 0\}| = c_0$, $|\{n \mid p(n) = 1\}| = c_1$, etc, and $k$ is the least witness of $|\{n \mid p(n) = k\}| = \infty$.
\end{definition}

Let $\mathrm{Sort}_k$ denote the restriction of $\mathrm{Sort}_*$ to $\{0,\ldots,k-1\}^\mathbb{N}$. $\mathrm{Sort}_2$ was introduced and studied in \cite{paulyneumann}, and then generalized to $\mathrm{Sort}_k$ in \cite{hoelzl2}.

\begin{proposition}
$\mathrm{Sort}_k \equivW \mathrm{Sort}_2^{k-1}$.
\begin{proof}
$\mathrm{Sort}_k \leqsW \mathrm{Sort}_2^{k-1}$: \quad Given a $\mathrm{Sort}_k$-instance $p$, we compute $k-1$ $\mathrm{Sort}_2$ instances $p_1, \ldots, p_{k-1}$ by letting $p_n$ be the result of replacing each digit less than $n$ in $p$ by $0$, and each digit greater-or-equal to $n$ by $1$. If $q_1,\ldots,q_{k-1}$ are the outputs of $\mathrm{Sort}_2$, we compute a suitable output $q$ of $\mathrm{Sort}_k$ by reading through all $q_i$ in lock-step. By construction, $q_i$ will switch from $0$ to $1$ prior to $q_{i+1}$ (if they switch at all). In $q$, we write $0$s as long as none switches, then $1$s after $q_1$ switched until $q_2$ switches, and so on.

$\mathrm{Sort}_2^{k-1} \leqW \mathrm{Sort}_{k}$: \quad Given $k-1$ $\mathrm{Sort}_2$-instances $p_1, \ldots, p_{k-1}$ we compute a $\mathrm{Sort}_k$-instance $p$. Since adding a digit of value $k-1$ to $p$ has no impact, we do so often enough to ensure an infinite sequence. We add the $\ell$-th digit $(k-1-j)$ to $p$ as soon as we find that $j$-many of the $p_i$ have an $\ell$-th digit $0$.

Let $q = \mathrm{Sort}_k(p)$. Each $0$ in $q$ indicates that all $p_i$ have one more $0$, so we compute the sorted versions of the $p_i$ by first copying the $0$s of $q$ to each $q_i$. If $q$ then changes from $0$ to $1$, this indicates that amongst the $p_i$ exactly one has no further $0$s. We can wait until we find another $0$ in $(k-2)$-many, and then continue the other $q_i$ by $1^\omega$. Then as long as we read $1$ in $q$, each remaining $q_i$ is extending by $0$. Once $q$ switches to $2$, we know than again exactly one $p_i$ has no further $0$s, and we identify the correct one by ing for $0$s in the others. (If $q$ jumps directly from $0$ to $2$, the two steps are just merged into one). We continue like this for the other potential switches in $q$.
\end{proof}
\end{proposition}

\begin{corollary}
\label{corr:sortvariants}
$\mathrm{Sort}_2^* \equivW \mathrm{Sort}_k^* \equivW \coprod_{k \in \mathbb{N}} \mathrm{Sort}_k$.
\end{corollary}

\begin{proposition}
$\mathrm{Sort}_2^* \leW \mathrm{Sort}_* \leW \lim$.
\begin{proof}
By Corollary \ref{corr:sortvariants}, we can show $\coprod_{k \in \mathbb{N}} \mathrm{Sort}_k \leqW \mathrm{Sort}_*$ for the first reduction. The reduction follows since each $\mathrm{Sort}_k$ is a restriction of $\mathrm{Sort}_*$. To show strictness, assume that $\mathrm{Sort}_* \leqW \mathrm{Sort}^*$. As $\mathrm{Sort}_*$ is easily seen to be a fractal, $\sigma$-irreducibility implies that there is some $n \in \mathbb{N}$ with $\mathrm{Sort}_* \leqW \mathrm{Sort}^n$, so in particular, $\mathrm{Sort}^{n+1} \leqW \mathrm{Sort}^n$. But this was shown to be false in \cite{paulyneumann}.

Using $\widehat{\lpo} \equivW \lim$, we can see $\mathrm{Sort}_* \leqW \lim$ by asking all countably many questions of the form \emph{Are there at least $n$ occurrences of $k$ in the input?} (each of them is equivalent to $\lpo$), and then computing the solution to $\mathrm{Sort}_*$ from that in the obvious way. To see that the reduction is strict, note that $\mathrm{Sort}_*$ outputs only computable points, whereas $\lim$ has to output the Halting problem for some computable input.
\end{proof}
\end{proposition}

\begin{proposition}
$\C_\mathbb{N} \star \mathrm{Sort}_* \equivW \Pi^0_2\mathrm{C}_\mathbb{N}$.
\begin{proof}
\begin{description}
\item[$\C_\mathbb{N} \leqW \Pi^0_2\C_\mathbb{N}$] Straightforward.
\item[$\mathrm{Sort}_* \leqW \Pi^0_2\C_\mathbb{N}$] For any $w \in \mathbb{N}^*$ the set $A_w \subseteq \Baire$ containing all $p$ that have exactly $w(k)$ many $k$'s and infinitely many $|w|$'s is a $\Pi^0_2$-set. Being in the domain of $\mathrm{Sort}_*$ means being contained in a $A_w$. With $\Pi^0_2\C_\mathbb{N}$ we can, given $p \in \dom(\mathrm{Sort}_*)$, find some $w$ such that $p \in A_w$. Then computing $\mathrm{Sort}_*(p)$ is straightforward.
\item[$\Pi^0_2\mathrm{C}_\mathbb{N} \star \Pi^0_2\mathrm{C}_\mathbb{N} \equivW \Pi^0_2\mathrm{C}_\mathbb{N}$] See \cite{paulybrattka5}.
\item[$\Pi^0_2\mathrm{C}_\mathbb{N} \leqW \C_\mathbb{N} \star \mathrm{Sort}_*$] We can view $\Pi^0_2\C_\mathbb{N}$ as the following task: \emph{Given a sequence $(p_i)_{i \in \mathbb{N}}$ with $p_i \in \Cantor$ such that some $p_n$ contains infinitely many $1$s, find such an $n$.} From such an input to $\Pi^0_2\mathrm{C}_\mathbb{N}$ we compute some $p \in \Baire$ by listing $n$ in $p$ whenever we find another $1$ in $p_n$. The promise on the input to $\Pi^0_2\mathrm{C}_\mathbb{N}$ ensures that $p \in \dom(\mathrm{Sort}_*)$. Now $\mathrm{Sort}_*(p)$ is an eventually constant sequence, and we can use $\C_\mathbb{N}$ to find the last index at which $\mathrm{Sort}_*(p)$ changes. With that index, we can find which digit is repeated infinitely often in $\mathrm{Sort}_*(p)$, and this digit constitutes a valid output to $\Pi^0_2\mathrm{C}_\mathbb{N}$.
    \end{description}
\end{proof}
\end{proposition}

\begin{proposition}
$\C_{\Cantor,\sharp<\infty} \leW \mathrm{Sort}_*$.
\begin{proof}
For each $k \in \mathbb{N}$, let $\langle w_n^{0,k},w_n^{1,k},\ldots,w_n^{k}\rangle$ be an effective enumeration of all prefix-independent $k+1$-tuples of finite words. Given $A \in \dom(\C_{\Cantor,\sharp<\infty})$, we will write the $N$-th $k$ to the input to $\mathrm{Sort}_*$ as soon as we have found for each $n \leq N$ some $w_n^{j,k}$ such that $w_n^{j,k} \cap A = \emptyset$ (this condition is semidecidable). If $|A| = K$, then for each $k \leq K$ we will eventually reach some $k$-tuple $\langle w_n^{0,k},w_n^{1,k},\ldots,w_n^{k}\rangle$ such that each $w_n^{j,k} \cap A$ contains a point, whereas each $K + 1$-tuple will eventually be enumerated. In particular, we do indeed produce a valid input to $\mathrm{Sort}_*$.

As $|A| > 0$, we know that the output $p$ of $\mathrm{Sort}_*$ is not $0^\omega$. Thus, we can compute the maximal $c_0 \in \mathbb{N}$ such that $0^{c_0}$ is a prefix of $p$. Then we know that $w_{c_0+1}^{0,0} \cap A \neq \emptyset$. We write $w_{c_0+1}^{0,0}$, and attempt to compute the unique $x \in w_{c_0+1}^{0,0} \cap A$, writing any confirmed prefix to the output. Simultaneously we search for the maximal $c_1$ such that $0^{c_0}1^{c_1}$ is a prefix of $p$.

If $|A| = 1$, then the process will proceed to fully write a correct solution. If $|A| > 1$, we will find a maximal $c_1$. Let $q_{\leq i_0}$ be the prefix of the output written so far (we know this to be correct). Then we know that $w_{c_1+1}^{0,1} \cap A \neq \emptyset$ and $w_{c_1+1}^{1,1} \cap A \neq \emptyset$. If $p_{\leq i_0}$ is compatible with one of $w_{c_1+1}^{0,1}$ and $w_{c_1+1}^{1,1}$, then we attempt to compute the unique extension of $p_{\leq i_0}$ in $w_{c_1+1}^{0,1} \cap A$ respectively $w_{c_1+1}^{1,1} \cap A$, while searching for a maximal $c_2$ such that $0^{c_0}1^{c_1}2^{c_2}$ is a prefix of $p$. If the process continues for ever, we correctly write a solution. Otherwise, some maximal $c_2$ must exist. If $p_{\leq i_0}$ is compatible with neither of $w_{c_1+1}^{0,1}$ and $w_{c_1+1}^{1,1}$, then $|A| \geq 3$, and we know that a maximal $c_2$ exist, which we search for straight-away.

We continue this process. At the latest once we have found the maximal $c_{|A|-1}$, we will attempting to compute the unique point in a singleton, and thus the computation will succeed.

That the reduction is strict follows from $\C_\mathbb{N} \leqW \mathrm{Sort}_2 \leqW \mathrm{Sort}_*$, $\C_{\Cantor,\sharp<\infty} \leqW \C_\Cantor$ and $\C_\mathbb{N} \nleqW \C_\Cantor$.
\end{proof}
\end{proposition}

\section*{Acknowledgements}
We are grateful to Matthew de Brecht for helpful discussions.

 \begin{minipage}{0.1\textwidth}\includegraphics[width=\textwidth]{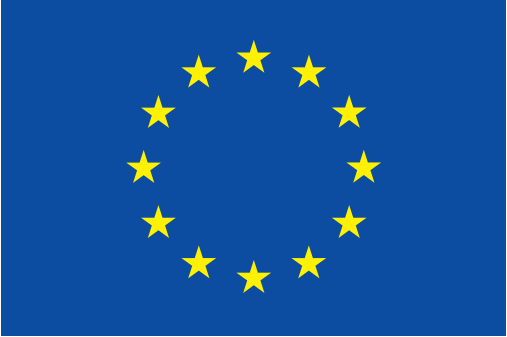}\end{minipage} \begin{minipage}{0.9\textwidth} This project has received funding from the European Union's Horizon 2020 research and innovation programme under the Marie Sklodowska-Curie grant agreement No 731143, \emph{Computing with Infinite Data}.\end{minipage}

 This work also benefited from the Marie Curie International Research Staff Exchange Scheme \emph{Computable
Analysis} (PIRSES-GA-2011- 294962).
 The second author was partially supported by the JSPS Core-to-Core Program (A.
Advanced research Networks) and JSPS KAKENHI Grant Number 15K00015.

\bibliographystyle{eptcs}
\bibliography{../spieltheorie}

\end{document}